\documentclass[a4paper,fleqn]{cas-dc}

\usepackage{algorithm}
\usepackage{subcaption}
\usepackage{tikz}
\usetikzlibrary{arrows.meta}
\usepackage{marginnote}
\usepackage{amsmath}
\usepackage{amssymb}
\usepackage[authoryear,longnamesfirst]{natbib}

\def\d{\,\mathrm d}	
\def\x{\,\mathrm x}
\usepackage{algpseudocode}
\usepackage{subcaption}
\newcommand{\mrm}[1]{\mathrm{#1}}

\newcommand{\mF}{\mrm{F}}

\newcommand{\bU}{\mathbf{U}}
\newcommand{\bUhat}{\hat{\mathbf{U}}}

\newcommand{\bN}{\mathbf{N}}
\newcommand{\be}{\mathbf{e}}

\newcommand{\Lp}[1]{L^{#1}}
\newcommand{\Hp}[1]{H^{#1}}
\newcommand{\Wp}[2]{W^{#1,#2}}
\newcommand{\dx}{\,\mathrm{d}\mathbf{x}}
\newcommand{\ds}{\,\mathrm{d}s}
\newcommand{\dt}{\,\mathrm{d}t}
\newcommand{\cR}{\mathcal{R}}
\newcommand{\cL}{\mathcal{L}}
\newcommand{\cG}{\mathcal{G}}
\newcommand{\cX}{\mathcal{X}}

\newtheorem{theorem}{Theorem}[section]
\newtheorem{lemma}{Lemma}[section]
\newtheorem{remark}{Remark}[section]
\newtheorem{definition}{Definition}[section]
\newtheorem{corollary}{Corollary}[section]
\newtheorem{proposition}{Proposition}[section]

\newtheorem{Assumption and notation}{Assumptions and notations}
\newtheorem{General assumptions}{General assumptions}[section]
\newtheorem{proof}{Proof}

\newtheorem{Definition and notation}{Definition and notation}
\everymath{\displaystyle}
\def\tsc#1{\csdef{#1}{\textsc{\lowercase{#1}}\xspace}}
\tsc{WGM}
\tsc{QE}

\begin{document}
\let\WriteBookmarks\relax
\def\floatpagepagefraction{1}
\def\textpagefraction{.001}

\shorttitle{}    

\shortauthors{}  

\title [mode = title]{Mathematical Modeling of Cancer–Bacterial Therapy: Analysis and Numerical Simulation via Physics-Informed Neural Networks}  


%
\author[1,2]{Ayoub Farkane}
\fnmark[1]
\ead{ayoub.farkane@univ-lorraine.fr}
\credit{Conceptualization, Methodology, Mathematical analysis, Numerical simulations, Writing - original draft}

\affiliation[1]{organization={TICLab, International University of Rabat},
            city={Rabat},
            postcode={11100},
            country={Morocco}}

\affiliation[2]{organization={Université de Lorraine, CNRS, CRAN, F-54000 Nancy, France}
            }

\author[2,3,4]{David Lassounon}
\fnmark[1]
\ead{david.lassounon@univ-lorraine.fr}
\credit{Conceptualization, Mathematical analysis, Supervision, Writing - review \& editing}

\affiliation[3]{organization={Université de Lorraine, CNRS, IECL, F-54000, Nancy, France}
            }

\affiliation[4]{organization={Univ Rennes, INSA, CNRS, IRMAR-UMR 6625, F-35000},
            postcode={35708},
           city={Rennes},
            country={France}}

\fntext[1]{These authors contributed equally to this work.}
\cormark[1]
\cortext[1]{Corresponding author}

\begin{abstract}
Bacterial cancer therapy exploits anaerobic bacteria’s ability to target hypoxia tumor regions, yet the interactions among tumor growth, bacterial colonization, oxygen levels, immunosuppressive cytokines, and bacterial communication remain poorly quantified. We present a mathematical model of five coupled nonlinear reaction–diffusion equations in a two-dimensional tissue domain. We proved the global well-posedness of the model and identified its steady states to analyze stability.  Furthermore, a physics-informed neural network (PINN) solves the system without a mesh and without requiring extensive data. It provides convergence guarantees by combining residual stability and Sobolev approximation error bounds. This results in an overall error rate of $\mathcal{O}(n^{-2}\ln^{4}(n) + N^{-1/2}),$
which depends on the network width $n$ and the number of collocation points $N$. We conducted several numerical experiments, including predicting the tumor’s response to therapy. We also performed a sensitivity analysis of certain parameters. The results suggest that long-term therapeutic efficacy may require the maintenance of hypoxia regions in the tumor, or using bacteria that tolerate oxygen better, may be necessary for long-lasting tumor control.  
 \nocite{*}
\end{abstract}





\maketitle

\section{Introduction}
\label{sec:intro}

\subsection{Background and Motivation of PINNs }
\label{sec:pinn_background}

Physics-informed neural networks (PINNs) were introduced in work~\cite{raissi2019physics}. They reformulate solving partial differential equations (PDEs) as minimizing a composite loss functional. This functional encodes the governing equations, initial conditions, and boundary conditions as soft constraints on a neural network. The approach offers three structural advantages over classical numerical methods.
First, no spatial mesh is required: derivatives are computed exactly at any
point in the domain via automatic differentiation through the network
graph~\cite{paszke2019pytorch}.
Second, the method is fully unsupervised: no labeled PDE solution data are
needed, only collocation points drawn from the domain.
Third, the same network simultaneously represents the solution at all times
and all points, making temporal queries at arbitrary resolutions trivial
after training.

These properties make PINNs well suited for systems of coupled, nonlinear PDEs with complex reaction terms. Classical methods often struggle with such systems. They face issues like mesh stiffness, operator splitting errors, or sensitivity to time-step size.
However, the theoretical convergence of PINNs applied to nonlinear systems
of PDEs has remained largely open, with rigorous results confined to linear
or semi-nonlinear equations~\cite{mishra2023E,shin2020ON,de2021approximation,de2024error}.
The present paper addresses this gap for a five-species nonlinear reaction-diffusion system arising in mathematical biology.
\subsection{Biological Context of Bacterial Cancer Therapy}
\label{sec:bio_context}
In $2022$, in the world, $20$ million new cases of cancer were diagnosed, resulting in nearly 9.7  million deaths worldwide, according to the World Health Organization (WHO). By $2050$, the number of new cases is expected to reach 35 million, an increase of $77\%$ according to~\cite{soerjomataram2021planning}. As a result, cancer remains a major cause of death worldwide. This disease results from the uncontrolled proliferation of tumor cells. These cells escape the normal mechanisms regulating the cell cycle and apoptosis~\cite{hanahan2022}, thereby promoting their survival, growth, and resistance to treatment. This behavior is closely linked to complex interactions with the tumor microenvironment (TME).

TME is a dynamic coupled ecosystem that can involve five main species: the tumor, bacteria, oxygen concentration, immunosuppressive cytokines, and quorum-sensing (QS) signaling molecules. Their interactions form a network of positive and negative feedback loops (Figure \ref{diag}). The tumor cells consume oxygen, creating hypoxia niches that promote bacterial growth, particularly anaerobic bacteria. The bacteria produce QS signals that inhibit tumor proliferation. In addition,  tumor stimulate the production of immunosuppressive cytokines, which, in turn, limit bacterial proliferation.
Immunosuppressive Cytokines produced by tumor cells can inhibit bacterial spread in two ways. On one hand, they activate local immune cells, such as macrophages and neutrophils, which destroy bacteria. On the other hand, they modify the tumor microenvironment by increasing blood perfusion or oxygenation, making conditions less favorable for the anaerobic bacteria used in bacterial cancer therapy (BCT).

Today, BCT uses a characteristic of tumors (the presence of hypoxia areas where conventional therapies such as chemotherapy or radiotherapy are ineffective). In these areas, therapeutic anaerobic bacteria can  survive~\cite{mascheroni2020}.  For example, bacteria such as \emph{Salmonella} and \emph{Clostridium novyi}-NT colonize these hypoxia areas and are activated to destroy tumors via their diffusible molecular signals (quorum-sensing (QS)  effect)~\cite{howell2025mathematical,aganja2024}. QS molecules, \emph{acyl-homoserine lactones} (AHL), and \emph{autoinducer peptides} (AIP)~\cite{guo2025}, accumulate and are proportional to bacterial density, and are responsible for the destruction of cancer cells.

Modeling this dynamic system mathematically requires a nonlinear reaction-diffusion equations (PDE) to describe spatial heterogeneity among different species. This approach helps us understand interactions between tumors, oxygen, therapeutic bacteria, and tumor-induced cytokines. It also helps predict anti-cancer bacterial activity, which is the focus of our work.

\subsection{Existing Mathematical Modeling}
\label{sec:literature}
Several mathematical models of the tumor or immune system response and the role of hypoxia in bacterial therapy have been extensively studied in recent years (see e.g.  \cite{milotti2020oxygen,mascheroni2020,sadhu2023impact,lampropoulos2025spatio,nova2025tumor,alarabi2022mathematical,howell2025mathematical,geretovszky2025}, and the references therein).
In \cite{mascheroni2020}, a model of anaerobic bacterial infiltration into hypoxia areas and the effect of oxygen on their growth was studied, without accounting for the immune response or bacterial quorum sensing. Very recently, \cite{geretovszky2025} proposed an ordinary differential equation model of interactions among tumors, bacteria, and chemotherapy treatment, but did not account for the immune response or QS signaling. Furthermore, in \cite{howell2025mathematical}, an ordinary differential equation model of \emph{Salmonella}-based cancer therapies is presented. Their work has shown that bacterial cytotoxicity alone is insufficient to treat tumors and that activating the immune system is important for effective treatment. In their model, the spatial heterogeneity of the different species (tumors, bacteria, immune response) and also interactions with tumor-derived immunosuppressive cytokines in response to bacteria are neglected.
These often-overlooked biological phenomena may provide insight into the indirect effects of anti-cancer bacteria on tumor cells and into the tumor's response to their presence in its environment.
\subsection{Deep Learning  Approaches to Tumor Modeling}
The intersection of deep learning and mathematical oncology has led to a rapidly growing body of research. However, physics-informed methods for modeling complex, multi-species tumor microenvironment remain largely unexplored.  Early studies used PINNs in simple settings. For instance, in this work \cite{rodrigues2024using} applied PINNs to find parameters in classical scalar tumor growth ODE models, such as the Verhulst logistic and Montroll power-law equations. This work showed robustness to experimental noise. However, these models were purely temporal and spatially uniform, which limited their relevance for tumors influenced by spatial gradients of nutrients and signaling molecules. Other studies have focused on diagnostic applications instead of mechanistic modeling. The paper \cite{kandasamy2025optimized}  developed an attention-based multi-scale convolutional neural network optimized with hybrid multi-objective CAT algorithms, achieving a high  accuracy in classifying bone marrow cancer cells. While this approach is highly effective, it relies entirely on data and does not include governing equations. This makes it complementary to, but fundamentally different from, physics-based tumor modeling. This framework \cite{sun2025physics} a physics-informed U-Net for dose prediction in intensity-modulated radiation therapy. Here, physical knowledge was included through structured input engineering rather than explicit PDE residual constraints.   This shows that “physics-informed” methods cover a wide range of enforcement strategies. A more rigorous PDE-constrained formulation was later proposed in this paper \cite{zhang2025personalized}. This work combined PINNs with the Fisher–KPP reaction–diffusion equation. This study showed the clinical feasibility of PINNs for single-species tumor modeling.  Most recently,  the work  \cite{el2026tumor} proposed a digital twin framework integrated U-Net segmentation, Chan–Vese active contour refinement, and a PINN-based reaction–diffusion solver. 

Despite significant advances, one limitation remains. Most tumor response models based on PINN networks focus on a single cell population. As a result, they cannot account for interactions underlying emerging therapies. To address this, the present work proposes and analyzes a system of five coupled nonlinear reaction-diffusion on a tissue domain. This model describes tumor response to bacterial therapy, tumor oxygen consumption, tumor-secreted immunosuppressive cytokines, and quorum-sensing regulation. Unlike previous methods that treat the basic equations as soft constraints without formal guarantees, our work establishes the global well-posedness of the system via a mass-control structure. Furthermore, we identify three  biologically relevant steady-states. In addition, we formally provide the suitability analysis for PINN approximation for considered system, including approximation error estimates.

\subsection{Contributions}
\label{sec:contributions}

This paper makes the following four contributions.

\begin{enumerate}

\item \textbf{A coupled reaction-diffusion system for BCT.}
We formulate and described system~\eqref{couple}, coupling tumor density response with quorum-sensing inhibition, hypoxia-activated bacterial growth (by a inverse Michaelis–Menten function), vascular oxygen exchange, tumor-induced immunosuppressive cytokines
production in response to bacteria (see Section \ref{sec:model}). The system is the first to simultaneously integrate  these cross-interactions.
\item \textbf{Well-posedness and steady-states analysis.}
Under general assumptions \textit{\textbf{(H)}} on initial data and system parameters in \eqref{couple}, we prove the following: The global existence, uniqueness, regularity, and positivity of a bounded solution in an appropriate Sobolev space (see Theorem \ref{wp} and Proposition \ref{propReg}, Section \ref{MAsection}). We also provide a steady-state stability analysis and show that no Turing instability is induced by diffusion. Details are provided in Section \ref{MAsection}.
\item \textbf{PINN convergence analysis for our system.}
We establish convergence guarantees for the proposed PINN applied to the system~\eqref{couple}. Specifically, our analysis includes: (i) a Lipschitz bound on the nonlinear operator (Lemma~\ref{lem:lipschitz}); (ii) a stability theorem (Theorem~\ref{thm:stability}) relating approximation error to the network residual via energy estimates; (iii) an approximation theorem (Theorem~\ref{thm:approx}); and (iv) a generalization bound (Lemma~\ref{lem:mc}) for the discrete loss. Finally, Theorem~\ref{thm:main} combines these results to characterize how the method converges as a function of the network architecture and the total number of collocation points.
\item \textbf{Numerical simulations.}
We validate the proposed framework with numerical simulations. We conducted various experiments to analyses the different behaviors of solutions. They show that the network describes the complex nonlinear dynamics of the underlying biological system with high precision and does not rely on synthetic or assumed parameter values. Simulations confirm the expected interactions  between the system components, that are illustrated in Figure~\ref{diag}.
\end{enumerate}
The rest of the article is organized as follows. Section~\ref{sec:model} formulates the mathematical model~\eqref{couple}. 
Section~\ref{MAsection} presents the mathematical analysis, including the well-posedness theorem with a complete proof and the steady-state characterization. 
Section~\ref{sec:pinn_method} develops the PINN framework for the considered system. 
Section~\ref{sec:convergence} proves the convergence theorems for the applied PINNs and discusses parameter calibration. 
Section~\ref{sec:num} presents validation results from various numerical experiments. 
Finally, Section~\ref{sec:conclusion} discusses limitations and perspectives.
\section{Mathematical Modeling}\label{sec:model}
\subsection{Evolution of Tumor Density under Treatment}
Mathematical modeling of tumor evolution under therapy (chemotherapy, immunotherapy, radiotherapy, etc.) has been intensively studied in recent years (see e.g. \cite{EFTIMIE2021110739,belmiloudi:hal-01524395,jarrett2018mathematical,jarrett2018incorporating,alma991008704579706535,Lorenzo2021QuantitativeIV,LBH2,LBH1} and the references therein) . Its value lies particularly in understanding the dynamic mechanisms of growth, resistance, and response to treatment. The tumor response model adopted here describes the indirect action of anticancer bacteria in tumor cells density $T$ in the tissue via diffusible molecular signals $S$ released by bacteria. We consider the following reaction diffusion equation:
\begin{align}\label{T}
\hspace*{-1cm}&\frac{\partial T}{\partial t}(t,x)=D_{T}\Delta T(t,x) + \rho_{T }T(t,x) \left(\!\!1 - \frac{T(t,x)}{c}\!\!\right )\nonumber\\&- \delta_{T} T(t,x) - \alpha_{S} S(t,x) T(t,x) \mbox{ in } \mathcal{Q}\!=\!(0, \tau)\times\Omega.
\end{align}
In equation \eqref{T}, $D_T$ is the diffusion coefficient of tumor cells in tissue. The second term in the second member of the equation represents the standard term describing tumor cell proliferation, characterized by an intrinsic proliferation rate $\rho_{T}$ and limited by the maximum capacity of the microenvironment $c$ \cite{howell2025mathematical}. The term $-\delta_{T}T$ is natural mortality due to apoptosis or necrosis. The last term $-\alpha_{S}ST$ reflects the influence of the diffusible bacterial signal $S$ on the tumor population. It describe the indirect and regulated cytotoxicity of bacteria. The higher the concentration of $S$, the greater the inhibition of tumor cells. The coefficient $\delta_{T}$ is the degradation rate, and $\alpha_{S}$ represents the efficacy of the cytotoxic action of the diffusible bacterial signal $S$. The set $\Omega\subset\mathbb{R}^{d}$ ($d=1,2,3$) is open bounded represents a tissue, whose boundary $\partial\Omega$ is sufficiently regular. The parameter $\tau>0$ is the final time horizon. The quantity $T(t, x)$ is the tumor density at time $t\in (0, \tau)$ and position $x\in\Omega$.
\subsection{Evolution of Bacterial Density}
The evolution of bacterial density $B$ in tumor tissue is modeled by taking into account their spatial diffusion, their proliferation controlled by oxygen availability $O$, their natural mortality, and their elimination by immunosuppressive cytokines $I$ present in the microenvironment. We consider the following equation:
\begin{align}\label{B}
\frac{\partial B}{\partial t}(t, x)&= D_{B} \Delta B(t, x) + \rho_{B} B(t, x) \frac{K_{H}}{K_{H} + O(t, x)} \nonumber\\&- \delta_{B }B(t, x) - \beta_{I}I(t, x) B(t, x) \mbox{ in } \mathcal{Q}.
\end{align}
The model \eqref{B} combines several key biological processes. The quantity $B(t,x)$ is the bacterial density at time $t\in (0, \tau)$ and position $x\in\Omega$. The term $D_{B} \Delta B$ represents the spatial diffusion of bacteria in tumor tissue. The coefficient $\rho_{B}$ is the bacterial proliferation rate. The bacterial growth is described by the term $\rho_{B} B {K_H}/{(K_{H} + O(t,x))}$. In this term, the quantity ${K_H}/{(K_{H} + O(t,x))}$ reflects the preference of bacteria for hypoxia conditions. In areas with high oxygen concentrations, there is low bacterial growth. For example, therapeutic bacteria such as \emph{Salmonella} \emph{Typhimurium} or \emph{Clostridium} \emph{novyi}-NT germinate and proliferate in hypoxia regions \cite{zhang2024attenuated,sharafabad2024ability}.
 The term $-\delta_{B}B$ corresponds to the natural mortality of bacteria with the mortality rate $\delta_{B}$. The term $-\beta_{I}I B$ models bacterial elimination by immunosuppressive cytokines $I(x,t)$ with $\beta_{I}$ the efficiency of the immunosuppressive response. The constant $D_{B} $ is the diffusion coefficient of bacteria, and  $\rho_{B}$ is the bacterial proliferation rate. The constant $K_H$ is the oxygen inhibition constant.
 \subsection{Evolution of Oxygen Density}
 The evolution of oxygen concentration, denoted as $O(x,t)$, in tissues is described by a diffusion-consumption model \eqref{O} that includes an external oxygen \cite{mascheroni2020,milotti2020oxygen}. This model reflects how oxygen diffuses through tissue, is consumed by tumor cells, and is replenished from the surrounding vascular environment. It provides a framework for representing the spatiotemporal dynamics of hypoxia, which is a crucial factor in tumor growth and bacterial colonization. For the oxygen concentration in tissues, we consider the following equation: 
 \begin{align}\label{O}
\frac{\partial O}{\partial t}&=D_{O}\Delta O(t,x)-\gamma_{T} T(t,x)O(t,x)\nonumber\\&+ \gamma_{vas}(O_{vas}-O(t,x)) \mbox{ in } \mathcal{Q}.
\end{align}
The equation \eqref{O} describes the classical processes that govern oxygen distribution within the tumor region. The term $D_{O} \Delta O$ represents the spatial diffusion of oxygen through the tissue, where $D_{O}$ is the diffusion coefficient that depends on the medium's structure and its porosity. 
The terms $-\gamma_{T}TO$ represents the consumption of oxygen by tumor cells. The multiplicative form of these terms indicates that oxygen consumption increases with cell density and local oxygen availability. The coefficients $\gamma_{T} $ represent the consumption rate.
The term $\gamma_{vas} (O_{vas} - O(t,x))$ accounts for the external oxygen supply from the vascular environment. The parameter $\gamma_{vas}$ controls the rate at which the local oxygen concentration approaches the reference value $O_{vas}$. The value of $O_{vas}$ represents the vascular irrigation. Including an explicit term representing internal vascular oxygen supply is essential for accurately modeling oxygen distribution within the tumor. It enables a proper balance between diffusion, cellular consumption, and the limited contribution of blood perfusion. Without such a term, oxygen gradients would be underestimated, resulting in an inadequate representation of hypoxia regions.
 \subsection{Evolution of Immunosuppressive Cytokines Density}
The evolution of the density of immunosuppressive cytokines, denoted $I(x,t)$, in the TME is described by a reaction-diffusion equation \eqref{I} inspired by the work of \cite{howell2025mathematical}. The equation describes the biological process by which tumor cells produce immunosuppressive cytokines to evade therapies. Secreted by tumor cells, these cytokines, such as TGF-$\beta$ and IL-$10$, diffuse into surrounding tissues to establish an immunosuppressive tumor microenvironment and naturally degrade. Together, they control inflammation, angiogenesis, and tumor cell survival \cite{terry2020hypoxia,quail2013, hanahan2022}. For example, cytokines IL-$10$ and TNF-$\alpha$ activate the NF-$\kappa$B signaling pathway, contributing to inflammation and tumor progression. Their evolution is described by the following equation
\begin{align}\label{I}
\hspace*{-0.5cm}\frac{\partial I}{\partial t}(t,x)= D_{I}\Delta I(t,x) + \beta_{T}T(t,x)-\delta_{I}I(t,x) \mbox{ in } \mathcal{Q}.
\end{align}
The model describes the distribution of tumor immunosuppressive cytokines in the tissue. The term $D_{I} \Delta I$ represents the spatial diffusion of cytokines through the tissue, where $D_{I}$ is the diffusion coefficient. The term $\beta_{T} T$ represents cytokines  production by tumor cells. This production is proportional to the local density of tumor cells and reflects tumor-induced immune activation. The term $-\delta_{I} I$ represents the spontaneous degradation of cytokines, taking into account their half-life and elimination by consumption by other immune cells. By combining these terms, the model allows for the consideration of spatial gradients and the temporal evolution of immunosuppressive cytokines. This approach is essential for representing inhibitory effects on bacterial growth. Consequently, \eqref{I} enables the investigation of immune feedback on both bacteria and tumor.
 \subsection{Evolution of Diffusible Bacterial Signals}
 As explained in the first part of the introduction, the integration of diffusible bacterial molecular signals represented by $S$ acting on tumor dynamics allows us to understand microbial communication, such as the \emph{quorum}-\emph{sensing} (\cite{balhouse2017n}) on the viability of cancer cells. These signals can be \emph{acyl}-\emph{homoserine lactones} (AHLs) or \emph{autoinducing peptides} (AIPs) \cite{guo2025}. The dynamics of the diffusible bacterial signal $S(x,t)$ in time and space are described by a diffusion-production-degradation model given by the equation \eqref{S}. This model shows how the signal diffuses through tissue, is produced by bacteria, and degrades naturally. It shows the influence of bacterial signals of tumor viability carried by the term $-\alpha_{S}ST$ in the tumor equation \eqref{T}. We consider the following equation:
\begin{align}\label{S}
\hspace*{-0.8cm}\frac{\partial S}{\partial t}(t,x)= D_{S}\Delta S(t,x) + \beta_{B}B(t,x)-\delta_{S}S(t,x) \mbox{ in } \mathcal{Q}.
\end{align}
The term $D_{S} \Delta S$ refers to the spatial diffusion of the signal within the tissue, where $D_{S}$ is the diffusion coefficient. The term $\beta_{B} B$ models the production of the signal by bacteria, with $\beta_{B}$ the production parameter. Meanwhile, the term $-\delta_{S}S$  accounts for the natural degradation of the signal in the tissue, where $\delta_{S}$ is the degradation coefficient.
The \emph{quorum}-\emph{sensing} mechanism plays an essential role in the collective coordination of bacteria within the tumor microenvironment. Bacteria produce and release a diffusible molecular signal into the environment, the concentration of which increases with bacterial density. 
When this concentration exceeds a certain threshold, it triggers the expression of specific genes linked to virulence, cytotoxicity, or the production of anti-cancer enzymes. This approach makes \eqref{S} more consistent with experimental observations of bacterial therapies using \emph{quorum}-\emph{sensing} strains (e.g., \emph{Salmonella} in \cite{howell2025mathematical} and \emph{Pseudomonas} in \cite{pang2022pseudomonas}). In the literature, mathematical modeling of molecular signals of bacteria is often neglected. Their explicit inclusion highlights their diffuse local interactions with the tumor.\\
The Figure \ref{diag} presents a simplified diagram of the main interactions between variables. The arrows indicate the influences (positive in green, negative in red) deduced from the reaction terms of the equations (without diffusion). The nodes represent only the variables $T$ (tumor), $B$ (bacteria), $O$ (oxygen), $I$ (immunosuppressive cytokines), $S$ (signal). Tumor cells consume oxygen ($T$ $\to$ $O$), which promotes hypoxia and stimulates bacterial growth ($O$ $\to$ $B$). Bacteria produce a cytotoxic diffusible signal ($B$ $\to$ $S$ $\to$ $T$) that inhibits tumor proliferation. The tumor also induces the production of immunosuppressive cytokines ($T$ $\to$ $I$), which exert an inhibitory effect on bacteria ($I$ $\to$ $B$), thus forming a network of negative and positive feedback loops. This interaction describes the combined effects of hypoxia, bacterial cytotoxicity, and the immune response on the tumor.

In this paper, we will consider the following initial conditions, e.g., respectively the initial tumor density, bacterial density, oxigen concentration, immunosuppressive cytokines concentration, and the diffusible bacterial signals at initial time $t=0$ and for all $x\in\Omega$:
\begin{align}\label{initial}
\hspace*{-1cm}&T(0,x)=T_{0}(x), B(0,x)=B_{0}(x), O(0,x)=O_{0}(x), \nonumber\\&I(0,x)=I_{0}(x), S(0,x)=S_{0}(x) \mbox{ in } \Omega,
\end{align}
with homogeneous Neumann boundary conditions as follows
\begin{equation}\label{bord}
\hspace*{-1cm}\displaystyle{\frac{\partial T}{\partial n}=\frac{\partial B}{\partial n}=\frac{\partial O}{\partial n}=\frac{\partial I}{\partial n}=\frac{\partial S}{\partial n}=0 \mbox{ on } \Sigma:=(0, \tau)\times\partial\Omega.}
\end{equation}
where the vector $n$ is the outward normal to $\partial\Omega$.\\
By combining equations \eqref{T}-\eqref{B}-\eqref{O}-\eqref{I}-\eqref{S}-\eqref{initial}-\eqref{bord}, we obtain the following interactive system of diffusion reaction equations:
\begin{equation}\label{couple}
\hspace*{-0.8cm}\begin{cases}
\frac{\partial T}{\partial t}(t,x)& =D_{T}\Delta T(t,x) + \rho_{T }T(t,x) \left(1 - \frac{T(t,x)}{c}\right)\\&- \delta_{T} T(t,x) - \alpha_{S} S(t,x) T(t,x) \mbox{ in } \mathcal{Q},
\\[0.4em]
\displaystyle
\frac{\partial B}{\partial t}(t, x)& = D_{B} \Delta B(t, x) + \rho_{B} B(t, x) \frac{K_{H}}{K_{H} + O(t, x)}\\& - \delta_{B }B(t, x) - \beta_{I}I(t, x) B(t, x) \mbox{ in } \mathcal{Q},
\\[0.4em]
\displaystyle
\frac{\partial O}{\partial t}(t,x)& =D_{O}\Delta O(t,x)-\gamma_{T} T(t,x)O(t,x) \\&+ \gamma_{vas}(O_{vas}-O(t,x)) \mbox{ in } \mathcal{Q}, 
\\[0.4em]
\displaystyle
\frac{\partial I}{\partial t}(t,x)& = D_{I}\Delta I(t,x) + \beta_{T}T(t,x)-\delta_{I}I(t,x) \mbox{ in } \mathcal{Q},
\\[0.4em]
\displaystyle
\frac{\partial S}{\partial t}(t,x)& = D_{S}\Delta S(t,x) + \beta_{B}B(t,x)-\delta_{S}S(t,x) \mbox{ in } \mathcal{Q},
\\[0.4em]
&\displaystyle{\frac{\partial T}{\partial n}=\frac{\partial B}{\partial n}=\frac{\partial O}{\partial n}=\frac{\partial I}{\partial n}=\frac{\partial S}{\partial n}=0 \mbox{ on } \Sigma,}
\\[0.4em]
\displaystyle
T(0,x)&= T_{0}(x), B(0,x)=B_{0}(x), O(0,x)=O_{0}(x),\\&I(0,x)=I_{0}(x), S(0,x)=S_{0}(x) \mbox{ in } \Omega.
\end{cases}
\end{equation}
The system \eqref{couple} describes the simultaneous evolution of biological variables $T$,$B$,$O$,$I$,$S$, representing tumor density, bacterial density, oxygen concentration, immune cytokine concentration, and diffusible signal concentration, respectively. It formalizes the cross-feedback loops that govern the balance between tumor cell proliferation and destruction, particularly under the effect of hypoxia and bacterial \emph{quorum}-\emph{sensing} mechanisms. This coupling provides a better representation of the complex dynamics of the tumor microenvironment under bacterial therapy. The mathematical analysis of system \eqref{couple} will be discussed in the following section.
\section{Mathematical Analysis}\label{MAsection}
\subsection{Existence and Uniqueness of Global Classical Solutions}\label{sec:math}
In this section, we state and prove the existence, uniqueness, and regularity of global classical solutions of the system \eqref{couple}. 
Before that, we make the following assumptions, which we assume hold throughout the paper.
\begin{General assumptions}
\textit{ $\textbf{(H)}:$
We assume that the initial data $(T_{0},B_{0},O_{0},I_{0},S_{0})=:\mathbf{U}_{0} \in(L^{\infty}(\Omega)_{+}\cap H^{1}(\Omega))^{5}$. The real diffusion coefficients  $D_T, D_B, D_O, D_I, D_S$ are non-negative, and 
  the real proliferation, degradation, and consumption coefficients $\rho_{T}, \rho_{B}, \beta_{B}, \beta_{T}, \beta_{I}, \delta_{T}, \delta_{B}, \delta_{I}, \delta_{S},\\
  \alpha_{S}, \gamma_{T}, \gamma_{vas}, K_{H},O_{vas}$ are non-negative.
  }
  \end{General assumptions}
\begin{figure}
    \centering
\includegraphics[width=0.7\linewidth]{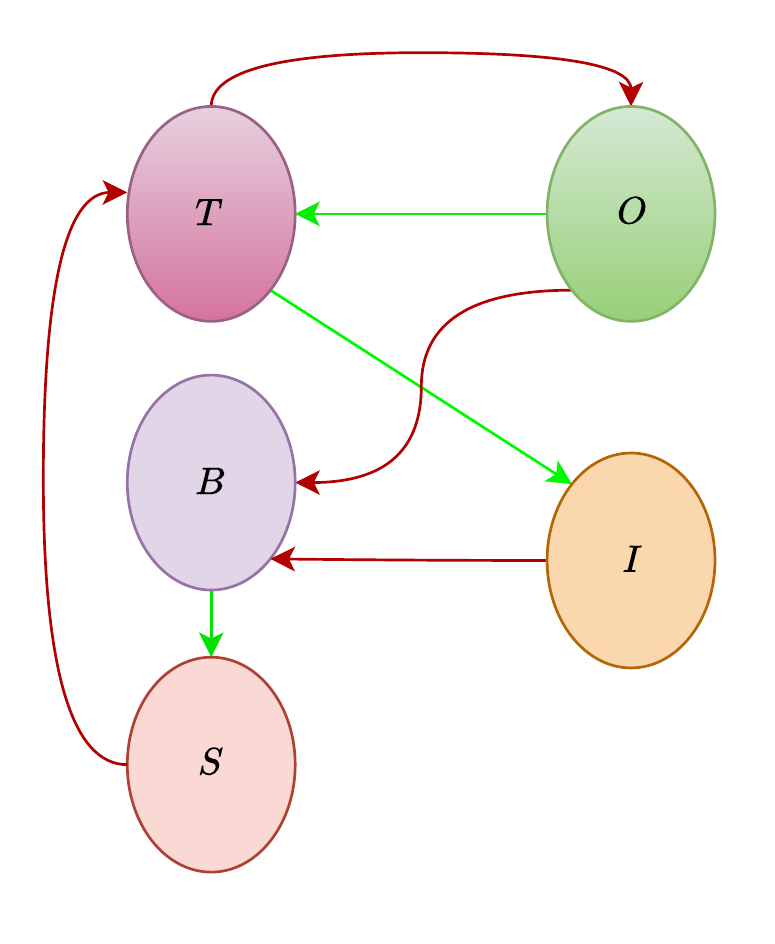}
\caption{Simplified diagram of the main interactions variables $T,B,O,I,S$.}
\label{diag}
\end{figure}
Throughout, we denote by $C$ a generic positive constant whose exact value is not important and may vary from line to line.\\
Now, we collect some known definitions of a classical solution of \eqref{couple}. For more details, we refer the reader to \cite{pierre2010global}.
\begin{definition}\label{defcs}
The classical solution to \eqref{couple} on $[0,\tau)$,
is a function
$
\mathbf{U}:=(T,B,O,I,S)^{\top} : [0,\tau) \times \Omega \to \mathbb{R}^5
$
such that
$$
\mathbf{U} \in C\big([0,\tau);(L^1(\Omega))^5\big)
\cap (L^\infty\big([0,\tau-\tau^{*}]\times\Omega\big))^5,
 \forall \tau^{*} \in (0,\tau),
$$
and for all $k,l=1,\dots,N$ and all $p\in[1,\infty)$,
$$
\partial_t \mathbf{U},\ \partial_{x_k}\mathbf{U},\ \partial_{x_k x_l}\mathbf{U}\in (L^p\big((\tau^{*},\tau-\tau^{*})\times\Omega\big))^5,
$$
and the equations in \eqref{coupleAd} are satisfied a.e..
\end{definition}
Let the function $\mF= (F_1,\dots,F_5)\in \mathbb{R}^5$.
We say that $\mF$ is \emph{quasi-positive} if it satisfies the following property:
\begin{align}\label{P}
\forall r \in ([0,+&\infty))^5,\ \forall i = 1,\dots,5,\nonumber\\
&F_{i}(r_1,\dots,r_{i-1},0,r_{i+1},\dots,r_m) \ge 0,
\end{align}
where we denote $r = (r_1,\dots,r_5)$.\\
We say that $\mF$ satisfies a \emph{mass-control structure} if there exists a constant $C>0$, a vector $\mathbf{b}\in\mathbb{R}^{5}$, and
a lower triangular invertible $5\times 5$ matrix $\mathbf{P}$ with nonnegative entries such that
\begin{equation}\label{M}
\forall r \in ([0,+\infty))^5, \quad
\mathbf{P}\mF(r)
\le C \left[ 1 + \sum_{i=1}^5 r_i \right]\mathbf{b}.
\end{equation}
\begin{lemma}\label{MP}
  The nonlinearity $ \mF=(F_{T},F_{B},F_{O},F_{I},F_{S})^{\top}$ from \eqref{couple} with $U=(T,B,O,I,S)^{ \intercal}\ge0$ componentwise and
\begin{align}\label{F}
\begin{split}
F_{T}({U})&=\rho_{T }T\left(1 - \frac{T}{c}\right)-\delta_{T} T-\alpha_{S} S T;\\
  F_{B}({U})&=\rho_{B}B\frac{K_{H}}{K_{H}+O}-\delta_{B}B-\beta_{I}IB;\\
  F_{O}({U})&=-\gamma_{T} T O +\gamma_{vas} (O_{vas}-O); \\
  F_{I}({U})&= \beta_{T}T-\delta_{I}I;\\
 F_{S}(U)&= \beta_{B}B-\delta_{S}S,
 \end{split}
   \end{align} 
   is quasi-positive and satisfies the mass-control structure.
\end{lemma}
\begin{proof}
Let $U=(T,B,O,I,S)^{ \intercal}\ge0$ componentwise, we have
\begin{gather}
\begin{split}
 F_{T}(0,B,O,I,S)&=0\ge 0;\\
  F_{B}(T,0,O,I,S)&= 0\ge 0;\\
  F_{O}(T,B,0,I,S)&= \gamma_{vas}O_{vas}\ge 0;  \\
  F_{I}(T,B,O,0,S)&= \beta_{T}T\ge 0;\\
 F_{S}(T,B,O,I,0)&= \beta_{B}B\ge 0.
 \end{split}
   \end{gather}
   Then, according to \eqref{P}, $\mF=(F_{T},F_{B},F_{O},F_{I},F_{S})^{\top}$ is quasi-positive. Moreover, for all $U=(T,B,O,I,S)^{ \intercal}\ge0$ componentewise, we have ${K_{H}}/({K_{H}+O})\leq 1$ and that 
   \begin{gather}\label{mcs}
       \begin{split}
     F_{T}(U)&+F_{B}(U)+F_{O}(U)+F_{I}(U)+F_{S}(U)\\
     &\leq \rho_{T}T+\rho_{B}B+\gamma_{vas}O_{vas}+\beta_{T}T+\rho_{B}B
     \\
    & \leq C(1+T+B)\\
    & \leq C(1+T+B+O+I+S),
       \end{split}
   \end{gather}
   implies that $\mF=(F_{T},F_{B},F_{O},F_{I},F_{S})^{\top}$ satisfies the mass-control structure in the sense of \eqref{M}  for  $\mathbf{P}=\mathrm{diag}(1,1,1,1,1)$ and $\mathbf{b}=(1,1,1,1,1)^{\top}$.
\end{proof}
The existence of global classical solutions of \eqref{couple} is given by the following 
\begin{theorem}\label{wp}
    The system \eqref{couple} has a unique global classical solution $\mathbf{U}=(T,B,O,I,S)^{\top}\in (L^{\infty}(\mathcal{Q})_{+})^{5}$ corresponding to the initial data $\mathbf{U}_{0}=(T_{0},B_{0},O_{0},I_{0},S_{0})^{\top}\in (L^{\infty}(\Omega)_{+})^{5}$ in the sense of Definition \ref{defcs}.
\end{theorem}
\begin{proof}
The global well-posedness of \eqref{couple} is demonstrated using a more general framework presented in \cite[Lemma 1.1 and Theorem 3.5]{pierre2010global}.
Before applying this result, we prove that our system \eqref{couple} fits well within this framework by showing that for all $U=(T,B,O,I,S)^{\top}\ge0$ componentwise, the nonlinear operator  $\mF(U)$ satisfies \cite[conditions (P) and (M)]{pierre2010global} given by conditions \eqref{P} and \eqref{M}. Indeed, the nonlinearity $\mF$ has at most polynomial and of class $C^{1}(([0,\infty))^{5};\mathbf{R}^{5})$. By Lemma \eqref{MP}, $\mF$ satisfies the quasi-positivity condition and admits a mass-control structure. Then, by  \cite[Lemma 1.1 and Theorem 3.5]{pierre2010global}, for any initial data $\mathbf{U}_{0}=(T_{0}, B_{0},I_{0},S_{0})^{\top}\in(L^{\infty}(\Omega)_{+})^{5},$ the system \eqref{couple} admits a unique classical global solution $\mathbf{U}=(T,B,O,I,S)^{\top}\in (L^{\infty}(\mathcal{Q})_{+})^{5}$.
\end{proof}
Since from $\textbf{\textit{(H)}}$, we have more regularity on the initial data $\mathbf{U}_{0}$, we then demonstrate that the unique corresponding global classical solution satisfies the regularities given in the following
\begin{proposition}\label{propReg}
 The global classical solution $\mathbf{U}$ to \eqref{couple} corresponding to the initial data $\mathbf{U}_{0}\in(L^{\infty}(\Omega)_{+}\cap H^{1}(\Omega))^{5}$ satisfies
 \begin{gather}
 \begin{split}
    &\mathbf{U}\in C([0,\tau];(L^{2}(\Omega))^{5})\cap (L^{\infty}(0,\tau; H^{1}(\Omega)))^{5},\\
  &  \mathbf{U}\in \left(L^{2}(0,\tau; H^{2}(\Omega))\cap H^{1}(0,\tau; L^{2}(\Omega))\right)^{5}.
    \end{split}
 \end{gather}
 such that for all $t\in (0, \tau)$,
 \begin{align}
\begin{split}
\lVert \mathbf{U}(t)\lVert^{2}_{(L^{2}(\Omega))^{5}} &+2d_{0}\int^{t}_{0}\lVert\mathbf{\nabla} \mathbf{U}(s)\lVert^{2}_{(L^{2}(\Omega))^{5}}\d s \\&\leq C\lVert \mathbf{U}_{0}\lVert^{2}_{(L^{2}(\Omega))^{5}},
\end{split}
\end{align}
\begin{align}
\lVert\mathbf{\nabla}\mathbf{U}(t)\lVert^{2}_{(L^{2}(\Omega))^{5}}&+d_{0}\int^{t}_{0}\lVert\mathbf{\Delta}\mathbf{U}(s)\lVert^{2}_{(L^{2}(\Omega))^{5})}\d s \nonumber\\&\leq C\lVert\mathbf{\nabla}\mathbf{U}_{0}\lVert^{2}_{(L^{2}(\Omega))^{5}},
\end{align}
where $d_{0}:=\min(D_{T}, D_{B},D_{O},D_{I}, D_{S})>0.$
\end{proposition}
\begin{proof}
Since $\mathbf{U}_{0}\in(L^{\infty}(\Omega)_{+}\cap H^{1}(\Omega))^{5}\subset (L^{\infty}(\Omega)_{+})^{5}$, then according to Theorem \ref{wp}, the system \eqref{couple} admits a unique global classical solution $\mathbf{U}\in (L^{\infty}(\mathcal{Q})_{+})^{5}$. Then, $\mathbf{U}$ verifies the following abstract form of \eqref{couple}.
\begin{equation}\label{Eqstate}
\begin{cases}
\frac{\partial \mathbf{U}}{\partial t} +\mathcal{\mathbf{A}} \mathbf{U}= \mF(\mathbf{U}) \mbox{ in } \mathcal{Q}, \\[0.3em]
\partial_{n} \mathbf{U}=0_{\mathbb{R}^{5}} \mbox{ in } \Sigma, \\[0.3em]
\mathbf{U}(0,x)=\mathbf{U}_{0}(x) \mbox{ in } \Omega,
\end{cases}
\end{equation}
where $\mathcal{\mathbf{A}}=-\mathcal{\mathbf{D}}\mathbf{\Delta}$ with  $\mathcal{\mathbf{D}}=\mathrm{diag}(D_{T}, D_B, D_O, D_I,D_S)$, and the nonlinearity $\mF=(F_{T},F_{B},F_{O},F_{I},F_{S})$ define in \eqref{F}. We have:
\begin{align}\label{FU}
(\mF(\mathbf{U}),\mathbf{U})_{(L^2(\Omega))^5}
&=\! \int_{\Omega}\!\!F_{T}(\mathbf{U})T\d x+\!\int_{\Omega}\!\! F_{B}(\mathbf{U})B\d x\nonumber\\&+\!\int_{\Omega}\!\!F_{O}(\mathbf{U})O\d x+\!\int_{\Omega}\!\! F_{I}(\mathbf{U})I\d x\nonumber\\&+\!\int_{\Omega}\!\!F_{S}(\mathbf{U})S\d x.
\end{align}
We will estimate each of the five terms of \eqref{FU}.
For the first term, using the fact that $T\ge0$ since $\mathbf{U}\ge0$ componentwise, we have
\begin{gather}\label{te1}
\begin{split}
\hspace{-0.8cm}\int_{\Omega} F_T(\mathbf{U})T\d x&= \int_{\Omega}\rho_{T }T^{2}\left(1 - \frac{T}{c}\right)-\delta_{T} T^{2}-\alpha_{S} S T^{2}\d  x\\
&\leq C\lVert T\lVert^{2}_{L^{2}(\Omega)}.
\end{split}
\end{gather}
For the second term of \eqref{FU}, using $\mathbf{U}\ge0$ then, we have $K_{H}/(K_{H}+O)\leq 1$ and that
\begin{gather}\label{be1}
\begin{split}
\hspace{-1cm}\int_\Omega\!\!F_{B}(\mathbf{U})B\d x\!&=\!\!\! \int_\Omega\!\!\big(\rho_{B}B^{2}\frac{K_{H}}{K_{H}+O}-\delta_{B}B^{2}-\beta_{I}IB^{2}\big)\\
&\leq C\lVert B\lVert^{2}_{L^{2}(\Omega)}.
\end{split}
\end{gather}
Similarly to the first two estimates, using the fact that $\mathbf{U}\ge0$, and Hölder's and Young's inequality,  we obtain
\begin{gather}\label{oe1}
\begin{split}
\hspace{-1cm}\int_{\Omega} F_{O}(\mathbf{U})O\d x
&=\!\!\int_{\Omega}\!\!\big(\!\!-\gamma_{T} T O^{2} +\gamma_{vas} (O_{vas}O-O^{2})\big)\\
&\leq \gamma_{vas}\int_{\Omega}O_{vas}O\d x\\
&\leq C\big(1+\lVert O\lVert^{2}_{L^{2}(\Omega)}\big).
\end{split}
\end{gather}
For the penultimate term, applying Hölder's and Young's inequality, we have
\begin{gather}\label{ie1}
\begin{split}
\int_{\Omega}F_{I}(\mathbf{U})I\d x&= \int_{\Omega}\big(\beta_{T}TI-\delta_{I}I^{2}\big)\d x\\
&\leq \beta_{I}\lVert T\lVert_{L^{2}(\Omega)}\lVert I\lVert_{L^{2}(\Omega)}\\
&\leq C\big(\lVert T\lVert^{2}_{L^{2}(\Omega)}+\lVert I\lVert^{2}_{L^{2}(\Omega)}\big),
\end{split}
\end{gather}
and that
\begin{gather}\label{se1}
\begin{split}
\int_{\Omega} F_{S}(\mathbf{U})S\d x&=\int_{\Omega}\big(\beta_{B}BS-\delta_{S}S^{2}\big)\d x\\
\leq & C \big(\|B\|_{L^2(\Omega)}^2+\|S\|_{L^2(\Omega)}^2\big).
\end{split}
\end{gather}
By summing all estimates \eqref{te1}, \eqref{be1}, \eqref{oe1}, \eqref{ie1}, \eqref{se1}, we obtain
\begin{gather}\label{estF}
\begin{split}
(\mathcal \mF(\mathbf{U}),\mathbf{U})_{L^2(\Omega)^5}\leq  C\left(1+\lVert\mathbf{U}\|^{2}_{(L^2(\Omega))^{5}}\right).
\end{split}
\end{gather}
Now, taking the inner product of the fist equation of \eqref{Eqstate} by $\mathbf{U}$ and using Definition \ref{defcs}, we have
\begin{align}\label{aps}
\langle \partial_t \mathbf{U}(t), \mathbf{U}(t)\rangle
&+ \big(\mathcal{\mathbf{D}}\mathbf{\nabla} \mathbf{U}(t), \mathbf{\nabla}\mathbf{U}(t) \big)_{(L^{2}(\Omega))^{5}}\nonumber\\& = \big( \mF(\mathbf{U}(t)), \mathbf{U}(t) \big)_{(L^{2}(\Omega))^{5}},
\end{align}
where $\langle \cdot, \cdot \rangle$ denotes the duality bracket between
 $\big((H^{1}(\Omega)^{5})\big)'$ and $(H^{1}(\Omega))^{5}$, and $(\cdot,\cdot)_{(L^{2}(\Omega))^{5}}$ is the inner product in $(L^{2}(\Omega))^{5}$.\\
Setting $d_{0}:=\min(D_{T}, D_{B},D_{O},D_{I}, D_{S})>0$ integrating \eqref{aps} over $[0, t]$ for all $t\in(0,\tau)$:
\begin{align}
\frac{1}{2}\lVert \mathbf{U}(t)&\lVert^{2}_{(L^{2}(\Omega))^{5}}-\frac{1}{2}\lVert \mathbf{U}_{0}\lVert^{2}_{(L^{2}(\Omega))^{5}} \nonumber\\&+d_{0}\int^{t}_{0}\lVert\mathbf{\nabla} \mathbf{U}(s)\lVert^{2}_{(L^{2}(\Omega))^{5}}\d s \nonumber\\&\leq \int^{t}_{0}\big( \mF(\mathbf{U}(s)), \mathbf{U}(s) \big)_{(L^{2}(\Omega))^{5}}\d s
\end{align}
According to \eqref{estF}, we have for all $t\in(0,\tau):$
\begin{align}
\begin{split}
\frac{1}{2}\lVert &\mathbf{U}(t)\lVert^{2}_{(L^{2}(\Omega))^{5}}-\frac{1}{2}\lVert \mathbf{U}_{0}\lVert^{2}_{(L^{2}(\Omega))^{5}} \\
&+d_{0}\int^{t}_{0}\lVert\mathbf{\nabla} \mathbf{U}(s)\lVert^{2}_{(L^{2}(\Omega))^{5}}\d s \\
&\leq C\int^{t}_{0}\left(1+\|\mathbf{U}(s)\|^{2}_{(L^2(\Omega))^{5}}\right)\d s.
\end{split}
\end{align}
Form the Grönwall's Lemma, we have
\begin{align}\label{ls}
\begin{split}
\lVert \mathbf{U}(t)\lVert^{2}_{(L^{2}(\Omega))^{5}} &+2d_{0}\int^{t}_{0}\lVert\mathbf{\nabla} \mathbf{U}(s)\lVert^{2}_{(L^{2}(\Omega))^{5}}\d s \\&\leq C\lVert \mathbf{U}_{0}\lVert^{2}_{(L^{2}(\Omega))^{5}}.
\end{split}
\end{align}
Then,
\[
\mathbf{U}\in L^{\infty}(0,\tau;(L^{2}(\Omega))^{5})
\cap (L^{2}(0,\tau; H^{1}(\Omega)))^{5}.
\]

Moreover, since $
\mathbf{D}\mathbf{\Delta}\mathbf{U}
\in (L^{2}(0,\tau;(H^{1}(\Omega))') )^{5}$ \text{(because }
$\mathbf{U}\in (L^{2}(0,\tau; H^{1}(\Omega)))^{5})$
and $\mF(\mathbf{U})\in (L^{2}(0,\tau;(H^{1}(\Omega))') )^{5}$ (arguing as in \eqref{estF} with test functions in $(H^{1}(\Omega))^{5}$), then 
$
\partial_{t}\mathbf{U}=\mathbf{D}\mathbf{\Delta}\mathbf{U}+\mF(\mathbf{U})\in (L^{2}(0,\tau;(H^{1}(\Omega))') )^{5}$.
Therefore, according to Aubin–Lions Lemma \cite{L2002}, $\mathbf{U}\in C([0,\tau];(L^{2}(\Omega))^{5}).$

Thus,
\begin{gather}\label{reg1}
\mathbf{U}\in C([0,\tau];(L^{2}(\Omega))^{5})
\cap (L^{2}(0,\tau; H^{1}(\Omega)))^{5}.
\end{gather}
For the other results regularities of $\mathbf{U}$, let us first estimate $\mF(\mathbf{U})$ in the $(L^{2}(\Omega))^{5}$ norm. We note that
\begin{align}
\hspace*{-0.5cm}\|\mF(\mathbf{U})\|^{2}_{(L^2(\Omega))^5}
&=\|F_{T}(\mathbf{U})\|^{2}_{L^2(\Omega)}+\|F_{B}(\mathbf{U})\|^{2}_{L^2(\Omega)}\nonumber\\&+\|F_{O}(\mathbf{U})\|^{2}_{L^2(\Omega)}+\|F_{I}(\mathbf{U})\|^{2}_{L^2(\Omega)}\nonumber\\&+\|F_{S}(\mathbf{U})\|^{2}_{L^2(\Omega)}.
\end{align}
Taking the square term by term in \eqref{mcs}, and integrating over $\Omega$, we have
\begin{gather}\label{l2Fu}
\|\mF(\mathbf{U})\|^{2}_{(L^2(\Omega))^5}\leq C\Big(1+\|\mathbf{U}\|^{2}_{(L^2(\Omega))^5}\Big).
\end{gather}        
Taking the inner product of the first equation of \eqref{Eqstate} with $-\mathbf{\Delta}\mathbf{U}$, integrating over $\Omega$, using \eqref{l2Fu}, Hölder's, and 
Young's inequality, we obtain for all $\varepsilon>0$:
\begin{align}
\hspace*{0cm}\frac{1}{2} \frac{\d}{\d t}\nonumber&\|\mathbf{\nabla}\mathbf{U}\|^{2}_{(L^2(\Omega))^{5}} +d_{0}\|\mathbf{\Delta}\mathbf{U}\|^{2}_{(L^2(\Omega))^{5}}\nonumber\\&\leq\lVert\mF(\mathbf{U})\|_{(L^2(\Omega))^5}\lVert\mathbf{\Delta}\mathbf{U}\lVert_{(L^2(\Omega))^5}
\nonumber\\&\leq C_{\varepsilon}\lVert\mF(\mathbf{U})\|^{2}_{(L^2(\Omega))^5}+\varepsilon\lVert\mathbf{\Delta}\mathbf{U}\lVert^{2}_{(L^2(\Omega))^5}\nonumber\\&\leq C_{\varepsilon}\big(1+\|\mathbf{U}\|^{2}_{(L^{2}(\Omega))^5}\big)+\varepsilon\lVert\mathbf{\Delta}\mathbf{U}\lVert^{2}_{(L^2(\Omega))^5}.
\end{align} 
By choosing $\varepsilon=\frac{d_{0}}{2}$, and using a nonlinear Gr\"onwall Lemma, we have for all $t\in(0,\tau^{*})$:
\begin{align}\label{ls2}
\hspace{-1cm}\|\mathbf{\nabla}\mathbf{U}(t)\|^{2}_{(L^2(\Omega))^{5}}\!+\!d_{0}\!\!\int^{t}_{0}\!\!\!\!\Vert\mathbf{\Delta}\mathbf{U}(s)\lVert^{2}_{(L^2(\Omega))^{5}}\d s
\leq C\|\mathbf{\nabla}\mathbf{U}_{0}\|^{2}_{(L^2(\Omega))^{5}}.
\end{align}
This implies that 
\begin{gather}\label{reg2}
\mathbf{U}\in (L^{\infty}(0,\tau; H^{1}(\Omega))\cap L^{2}(0,\tau; H^{2}(\Omega)))^{5}.
\end{gather}
According to \eqref{l2Fu}, and since $\mathbf{U}\in(L^{2}(\mathcal{Q}))^{5}$, we have $\mF(\mathbf{U})\in (L^{2}(\mathcal{Q}))^{5}$. 
Moreover, since $\mathbf{U}\in (L^{2}(0,\tau;H^{2}(\Omega)))^{5}$, it follows that 
$\mathbf{D}\mathbf{\Delta}\mathbf{U}\in (L^{2}(\mathcal{Q}))^{5}$. 
Therefore, we obtain
\begin{gather}\label{reg3}
\partial_{t}\mathbf{U}=\mathbf{D}\mathbf{\Delta}\mathbf{U}+\mF(\mathbf{U})\in (L^{2}(\mathcal{Q}))^{5}.
\end{gather}
Therefore, the estimates and regularity results \eqref{reg1}, \eqref{reg2}, \eqref{reg3}, \eqref{ls}, and \eqref{ls2} complete the proof of Proposition \ref{propReg}.
\end{proof}
\begin{remark} 
    We used a homogeneous Neumann boundary condition for all variables $T,B,O,I,S$ in \eqref{bord} to maintain a clear mathematical and numerical framework for analyzing system \eqref{couple}. In our future work, we may consider other boundary conditions for some of the five variables, as in \cite{LBH2,LBH1}, where a nonlinear term is used at the boundary to model tumor invasion into other organs.
\end{remark}
\subsection{Nondimensionalization of the system}
In order to introduce dimensionless variables, we consider the following rescalings:
\begin{gather}\label{resca}
\begin{split}
\tilde{t} = \rho_T t; \tilde{\x} =\frac{\x}{L}; L=\sqrt{\frac{D_T}{\rho_T}}; \tilde T = \frac{T}{c}; \tilde B = \frac{B}{B_*}; \\
\tilde{O} = \frac{O}{O_{vas}}; \tilde I = \frac{I}{I_*}; \tilde S = \frac{S}{S_*}; I_* = \frac{\beta_T c}{\rho_T}; 
S_* = \frac{\beta_B B_*}{\rho_T}.
\end{split}
\end{gather}
By replacing $t, x, T, B, O, I, S$ obtained by \eqref{resca} in \eqref{couple}, we have in $\tilde\Omega=:\Omega/L$, the following dimensionless system (for simplicity, the tildes are omitted):
\begin{gather}
\label{coupleAd}
\begin{split}
\begin{cases}
\displaystyle
\partial_t  T
=
\Delta  T
+  T(1- T)
- t_0  T
- \alpha_0  S  T,
\\[0.4em]
\displaystyle
\partial_t  B
=
d_B \Delta  B
+ \mu_0  B
\frac{\kappa}{\kappa+ O}
- b_0  B
- \beta_0  I  B,
\\[0.4em]
\displaystyle
\partial_t  O
=
d_O \Delta  O
- \gamma_0  T  O
+ \varrho_0(1- O),
\\[0.4em]
\displaystyle
\partial_t  I
=
d_I \Delta  I
+  T
- \eta_0  I,
\\[0.4em]
\displaystyle
\partial_t  S
=
d_S \Delta  S
+  B
- \lambda_0  S,
\\[0.4em]
\displaystyle{
\partial_{n} T = \partial_{n} B = \partial_{n} O = \partial_{n} I = \partial_{n} S = 0 \mbox{ on } \Sigma,} 
\\[0.4em]
T(0,x)=T_{0}(x); B(0,x)=B_{0}(x);
\\[0.4em]
O(0,x)=O_{0}(x); I(0,x)=I_{0}(x); S(0,x)=S_{0}(x),
\end{cases}
\end{split}
\end{gather}
where the dimensionless parameters are given by
\begin{gather}
\begin{split}
&d_B = \frac{D_B}{D_T}; d_O = \frac{D_O}{D_T}; d_I = \frac{D_I}{D_T}; d_S = \frac{D_S}{D_T}; \mu_0 = \frac{\rho_B}{\rho_T};\\
& t_0 = \frac{\delta_T}{\rho_T};  b_0 = \frac{\delta_B}{\rho_T}; \eta_0 = \frac{\delta_I}{\rho_T}; \lambda_0 = \frac{\delta_S}{\rho_T};
\alpha_0 = \frac{\alpha_S S_*}{\rho_T};\\
 &\beta_0 = \frac{\beta_I I_*}{\rho_T}; 
\gamma_0 = \frac{\gamma_T c}{\rho_T};
 \varrho_0 = \frac{\gamma_{vas}}{\rho_T};
\kappa = \frac{K_H}{O_{vas}}.   
\end{split}
\end{gather}
\subsection{Homogeneous steady states analysis}\label{sec-steady}
In this subsection, we consider the system \eqref{coupleAd} omitting spatial dependence in order to study spatially homogeneous steady states. The Theorem~\ref{wp} guaranteed that the solution of system \eqref{coupleAd} is positive and bounded. The following proposition give the spatially homogeneous steady states and their positivity conditions 
\begin{proposition}\label{shs}
  The system \eqref{coupleAd} admits three spatially homogeneous stationary states given by:
  \begin{itemize}
      \item $E_{0}=(0,0,1,0,0)$ corresponding to the tumor-free equilibrium.
      \item If $t_{0}<1$,  then 
      $E_T \!= \!\left( \!1 \!- \!t_0,\,0,\,\frac{\varrho_0}{\varrho_0  \!+ \!\gamma_0(1 \!- \!t_0)},\,
\frac{1 \!- \!t_0}{\eta_0},\,0 \!\right)$ corresponds to the bacteria-free tumor equilibrium.
\item  If the following conditions
\begin{gather}\label{coec}
\begin{split}
&t_{0}<1;\\
&   b_0 \!< \mu_0 \frac{\kappa}{1\!+\!\kappa}; \\
&b_0 +\frac{\beta_0}{\eta_0} (1\!-t_0)\!>\!  \frac{\mu_0\kappa (\varrho_0 \!+\! \gamma_0 (1\!-\!t_0))}{\kappa (\varrho_0 \!+\! \gamma_0 (1\!-\!t_0)) \!+\! \varrho_0},  
\end{split}
\end{gather}
are met,\\ 
then
$E_{*}\!\!=\!\!\left(\!\!\bar{T}, \frac{\lambda_0}{\alpha_0} (1 \!-\! t_0 \!-\! \bar T), \frac{\varrho_0}{\varrho_0 \!+\! \gamma_0 \bar T}, \frac{\bar T}{\eta_0},  \frac{1}{\alpha_0}(1 \!-\! t_0 \!-\! \bar T)  \!\!\right)$ corresponds to the coexistence equilibrium, 
where $\bar T\in]0, 1-t_{0}[$ satisfies the quadratic equation:
\begin{equation}
    \mu_0 \frac{\kappa (\varrho_0 + \gamma_0 \bar T)}{\kappa (\varrho_0 + \gamma_0 \bar T) + \varrho_0} 
= b_0 + \frac{\beta_0}{\eta_0} \bar T.
\end{equation}
  \end{itemize}
\end{proposition}
\begin{proof}
A spatially homogeneous steady state $(\bar T,\bar B,\bar O,\bar I,\bar S)$ of \eqref{coupleAd} satisfies:
$$
\partial_t T=\partial_t B=\partial_t O=\partial_t I=\partial_t S=0, $$
$$
\Delta T=\Delta B=\Delta O=\Delta I=\Delta S=0.
$$
Thus, the steady states are solutions of the system
\begin{equation}\label{st}
\begin{cases}
0= \bar T(1-\bar T)-t_0\bar T-\alpha_0 \bar S \bar T, \\[0.3em]

0=\mu_0 \bar B \dfrac{\kappa}{\kappa+\bar O}-b_0\bar B-\beta_0 \bar I \bar B, \\[0.6em]

0=-\gamma_0 \bar T \bar O+\varrho_0(1-\bar O), \\[0.3em]

0=\bar T-\eta_0 \bar I, \\[0.3em]

0=\bar B-\lambda_0 \bar S .
\end{cases}
\end{equation}
From the last two equations we obtain
\begin{equation}\label{is}
\bar I=\frac{\bar T}{\eta_0}, 
\quad
\bar S=\frac{\bar B}{\lambda_0}.
\end{equation}
Moreover, in the third equation of \eqref{st} the oxygen density satisfies
\begin{equation}\label{ox}
\bar O=\frac{\varrho_0}{\varrho_0+\gamma_0 \bar T}.
\end{equation}
Substituting \eqref{is} into the first two equations of \eqref{st} gives
\begin{gather}
\label{tb1}
0=\bar T\left(1-\bar T-t_0-\alpha_0\frac{\bar B}{\lambda_0}\right),
\\
\label{tb2}
0=\bar B\left(\mu_0\frac{\kappa}{\kappa+\bar O}-b_0-\beta_0\frac{\bar T}{\eta_0}\right).
\end{gather}
We now identify the possible steady states.
\paragraph{ $\bullet$ Tumor-free equilibrium.}
If $\bar T=\bar B=0$, then according to \eqref{is} $\bar I=0$, $\bar S=0$ and from \eqref{ox} we obtain $\bar O=1$.  
Thus the tumor-free equilibrium is
$$
E_0=(0,0,1,0,0).
$$
\paragraph{$\bullet$ Bacteria-free tumor equilibrium.} If $\bar B=0$, then $\bar S=0$ and \eqref{tb1}-\eqref{tb2} becomes
$$
\bar T(1-\bar T-t_0)=0.
$$
Hence $\bar{T}>0$, if $t_0<1$ i.e., $\delta_{T}<\rho_{T}$, we obtain $\bar T=1-t_0.$\\
  From \eqref{is} and \eqref{ox}, the variables $\bar{I}, \bar{O}$ are
$$
\bar I=\frac{1-t_0}{\eta_0},
\quad
\bar O=\frac{\varrho_0}{\varrho_0+\gamma_0(1-t_0)}.
$$
Therefore the tumor equilibrium without bacteria is
$$
E_T=\left(1-t_0,\,0,\,\frac{\varrho_0}{\varrho_0+\gamma_0(1-t_0)},\,
\frac{1-t_0}{\eta_0},\,0\right).
$$
\paragraph{$\bullet$ Coexistence equilibrium.}If $\bar{T}>0$ and $\bar{B}>0$, from \eqref{tb1}, \eqref{is} and \eqref{ox} we obtain
\begin{equation}\label{tiso}
 \bar T=1-t_0-\alpha_0\frac{\bar B}{\lambda_0}; \bar I=\frac{\bar T}{\eta_0}; \bar S=\frac{\bar B}{\lambda_0}; \bar O=\frac{\varrho_0}{\varrho_0+\gamma_0 \bar T}.
\end{equation}
Then, according to \eqref{tb2} the coexistence equilibrium satisfy the condition
\begin{equation}\label{mt}
\mu_0\frac{\kappa}{\kappa+\bar O}=b_0+\beta_0\frac{\bar T}{\eta_0}.
\end{equation}
 We now determine the positive coexistence steady state $E_*=(\bar T, \bar B, \bar O, \bar I, \bar S)$. For that, suppose that conditions \eqref{coec} are met.\\
Substituting $\bar O={\varrho_0}/(\varrho_0 + \gamma_0 \bar T)$ in equation \eqref{mt} we have an quadratic equation for $\bar T$:
\begin{equation}\label{kt}
\mu_0 \frac{\kappa (\varrho_0 + \gamma_0 \bar T)}{\kappa (\varrho_0 + \gamma_0 \bar T) + \varrho_0} 
= b_0 + \frac{\beta_0}{\eta_0} \bar T.
\end{equation}
Let define the function $f$ in $\mathbb{R}_{+}$ such that
$$f(\xi)=\mu_0 \frac{\kappa (\varrho_0 + \gamma_0 \xi)}{\kappa (\varrho_0 + \gamma_0 \xi) + \varrho_0} -b_0 - \frac{\beta_0}{\eta_0}\xi.$$
Equation \eqref{kt} can be rewritten as
\begin{equation}\label{ktf}
f(\bar T) = 0.
\end{equation}
We have \begin{gather}
\begin{split}
f(0) &= \mu_0 \frac{\kappa}{1+\kappa} - b_0,\\
f(1-t_0) &= \mu_0 \frac{\kappa (\varrho_0 + \gamma_0 (1-t_0))}{\kappa (\varrho_0 + \gamma_0 (1-t_0)) + \varrho_0}\\
&- b_0 - \frac{\beta_0}{\eta_0} (1-t_0).
\end{split}
\end{gather}
According to conditions \eqref{coec}, we obtain $f(0)>0$ and $f(1-t_{0})<0$ and
since $f$ is continuous on $\mathbb{R}_{+}$, then the equation \eqref{ktf} has at least one solution $\bar{T}$ in $]0, 1-t_{0}[$ .
Once the existence of $\bar{T}\in]0, 1-t_{0}[$ is established, then from \eqref{tiso} and according to \eqref{coec} again we have
\begin{gather}\label{rst}
\begin{split}
&\bar B=\frac{\lambda_0}{\alpha_0} (1 - t_0 - \bar T)>0;\quad \bar I = \frac{\bar T}{\eta_0};\\ 
&\bar S  = \frac{1}{\alpha_0}(1 - t_0 - \bar T);\quad \bar O = \frac{\varrho_0}{\varrho_0 + \gamma_0 \bar T}.
\end{split}
\end{gather}
Therefore, we have the existence of the coexistence equilibrium $E_* =(\bar T,\bar B,\bar O,\bar I,\bar S)$ where $\bar T\in]0, 1-t_{0}[$ satisfies \eqref{kt} and $\bar B,\bar O,\bar I,\bar S$ are given by \eqref{rst}, which terminates the proof of Proposition \ref{shs}.
\end{proof}
\subsection{Stability analysis of the steady states}
We analyze the stability of the homogeneous system derived from \eqref{coupleAd}. This allows us to establish the conditions of the stability of equilibrium states, taking into account the formation of spatial patterns. Let $E=(\bar{T},\bar{B},\bar{O},\bar{I},\bar{S})$ be a homogeneous equilibrium described Proposition \ref{shs}. We study the local stability of the homogeneous equilibrium states $E$.
Let $\bU=(T,B,O,I,S)^\top$ solution of \eqref{coupleAd}. The system \eqref{coupleAd} can be rewritten in the abstract form \eqref{Eqstate} with
\begin{gather*}
   \partial_t \bU = \mathbf{D\Delta U} + \mF(\bU),
\end{gather*}
where $\mathbf{D}=\mathrm{diag}(1,d_B,d_O,d_I,d_S)$ and $\mF(\bU)$ the reaction operator terms in \eqref{coupleAd}. 

Now, by linearize the system around $E$ by writing $\bU=E+u,$ we have that the perturbation term $u$ satisfies
\begin{gather}\label{lineaz}
  \partial_t u = \mathbf{D\Delta u} + J(E)u,
\end{gather}
where $J(E)$ is the Jacobian matrix of $\mF$ evaluated at $E$ defined by
\begin{equation*}
J(E)=
\tiny{
\begin{pmatrix}
1-2\bar{T}-t_0-\alpha_0S & 0 & 0 & 0 & -\alpha_0\bar{T} \\
0 & \mu_0\frac{\kappa}{\kappa+\bar{O}}-b_0-\beta_0\bar{I} &
-\mu_0\bar{B}\frac{\kappa}{(\kappa+\bar{O})^2} &
-\beta_0\bar{B} & 0 \\
-\gamma_0\bar{O} & 0 & -\gamma_0\bar{T}-\varrho_0 & 0 & 0 \\
1 & 0 & 0 & -\eta_0 & 0 \\
0 & 1 & 0 & 0 & -\lambda_0
\end{pmatrix}.
}
\end{equation*}
Due to Neumann's boundary conditions in \eqref{coupleAd}, the following Laplacian eigenvalue problem
\begin{gather*}
   -{\Delta} \phi_k = \lambda_k \phi_k 
\end{gather*}
admits a sequence of positive eigenvalues $(\lambda_{k})_{k=1,\cdots,5}$ such that $\lambda_1=0 <\lambda_{k}$ for all $k=2,\cdots, 5$ which correspond respectively to scalar eigenfunctions $(\phi_k)_{k=1,\cdots,5}$.

We assume that the perturbation $u$ as in the form
\begin{gather*}
u(x,t)=\sum_{k} c_k(t)\phi_k(x),
\end{gather*}
where the function $c_{k}\in C^{1}([0,\tau];\mathbb{R}^{5})$, and substituting into the linearized equation \eqref{lineaz} gives for all $k=1,\cdots,5:$
\begin{gather}
\frac{\d c_k}{\d t}=\left(J(E)-\lambda_k \mathbf{D}\right)c_k.
\end{gather}
Therefore the stability of the equilibrium $E$ is determined by the eigenvalues of the matrices $M_{k}$, $k=1,\cdots,5$ such that:
\begin{equation}
M_k=J(E)-\lambda_k\mathbf{D}.
\end{equation}
We have the following results
\begin{proposition}
\begin{itemize}
\item The tumor-free equilibrium
$$
E_0=(0,0,1,0,0)
$$
is locally asymptotically stable if
\begin{equation}
  t_0>1
\quad \text{and} \quad
\mu_0\frac{\kappa}{\kappa+1}<b_0.  
\end{equation}
\item The tumor equilibrium
$$
E_T=
\left(
1-t_0,0,
\frac{\varrho_0}{\varrho_0+\gamma_0(1-t_0)},
\frac{1-t_0}{\eta_0},
0
\right)
$$
exists for $t_0<1$ and is locally asymptotically stable if
\begin{equation}\label{sst}
\mu_0\frac{\kappa}{\kappa+\bar O}<b_0+\beta_0\frac{1-t_0}{\eta_0},
\end{equation}
where
$\bar O=\frac{\varrho_0}{\varrho_0+\gamma_0(1-t_0)}.$
\item The coexistence equilibrium
$$E_*=(\bar T,\bar B,\bar O,\bar I,\bar S)$$ is locally asymptotically stable if there exists.
\end{itemize}
\end{proposition}
\begin{proof}
\begin{itemize}
    \item \emph{The tumor-free equilibrium $E_{0}=(0,0,1,0,0)$}. Assume that $t_0>1,$ and $\mu_0\frac{\kappa}{\kappa+1}<b_0.$ Evaluating $M_{k}$ at $E_0$ $k=1,\cdots,5$ gives
\begin{equation*}
\tiny{
M_{k}(E_0)=
\begin{pmatrix}
1-t_0 -\lambda_{k}&0&0&0&0\\
0&\mu_0\frac{\kappa}{\kappa+1}-b_0 -\lambda_{k}d_{B}&0&0&0\\
-\gamma_0&0&-\varrho_0-\lambda_{k}d_{O}&0&0\\
1&0&0&-\eta_0 -\lambda_{k}d_{I}&0\\
0&1&0&0&-\lambda_0 -\lambda_{k}d_{S}
\end{pmatrix}.
}
\end{equation*}
Since $\lambda_{k}\ge0$ and coefficients $d_B,d_O,d_I,d_S \ge0$, then all real eigenvalues of $M_{k}(E_{0})$ are negative. Hence $E_0$ is locally asymptotically stable.
\item \emph{Bacteria-free tumor equilibrium:} \\$E_T=(1\!-\!t_0,0,\frac{\varrho_0}{\varrho_0\!+\!\gamma_0(1\!-\!t_0)},\frac{1\!-\!t_0}{\eta_0},0)$. Assume that $t_0<1$ and the condition \eqref{sst} is satisfied. For $k=1,\cdots,5$ the eigenvalues of $M_{k}(E_{T})$ are given by
\begin{gather*}
\begin{split}
\xi_{1}&=-(1-t_{0})-\lambda_k \leq 0,\\
\xi_2&= \mu_{0}\frac{\kappa}{\kappa+\bar{O}}-b_{0}-\beta_{0}\bar{I}-\lambda_k d_B \leq 0,\\
\xi_3&=-\gamma_0 (1-t_{0})-\varrho_0-\lambda_k d_O \leq 0,\\
\xi_4&=-\eta_0 - \lambda_k d_I \leq 0,\\
\xi_5&=-\lambda_{0}-\lambda_k d_S \leq 0.
\end{split}
\end{gather*}
Then, all real eigenvalues of $M_{k}(E_{T})$ are negative.
Therefore $E_T$ is locally asymptotically stable under the condition \eqref{sst}.
\item \emph{The coexistence equilibrium} $E_*=(\bar T,\bar B,\bar O,\bar I,\bar S)$ satisfying \eqref{rst} and \eqref{kt}. After calculating all nonzero coefficients of $J(E_{*})$, the matrix $M_{k}(E_{*})$ is:
\begin{equation*}
\tiny{
M_k(E_{*}) =
\begin{pmatrix}
-\bar T-\lambda_k & 0 & 0 & 0 & -\alpha_0\bar T \\

0 & -\lambda_k d_B &
-\mu_0 \bar B \dfrac{\kappa}{(\kappa+\bar O)^2} &
-\beta_0 \bar B & 0 \\

-\gamma_0 \bar O & 0 &
-\gamma_0 \bar T - \varrho_0 - \lambda_k d_O & 0 & 0 \\

1 & 0 & 0 &
-\eta_0 - \lambda_k d_I & 0 \\

0 & 1 & 0 & 0 &
-\lambda_0 - \lambda_k d_S
\end{pmatrix}.
}
\end{equation*}
The eigenvalues of the matrix $M_k(E_*)$ are given by
\begin{gather}
    \begin{split}
        \xi_1 &= -\bar T - \lambda_k\leq 0, \\
\xi_2 &= -\eta_0 - \lambda_k d_I\leq 0, \\
\xi_3 &= -\lambda_0 - \lambda_k d_S\leq 0, \\
\xi_4 &= -\gamma_0 \bar T - \varrho_0 - \lambda_k d_O \leq 0, \\
\xi_5 &= -\lambda_k d_B \leq 0.
    \end{split}
\end{gather}
Then, all spatial perturbations decay exponentially and the coexistence equilibrium $E_*$ is locally asymptotically stable. 
For $k=0$, we have $\xi_5=0$, and thus  $E_{*}$ is non-hyperbolic for the homogeneous system in space, but remains stable with respect to spatial perturbations.
\end{itemize}
\end{proof}
\begin{remark}
  The diffusion introduces the terms $-\lambda_k \mathbf{D}$, which shift the eigenvalues toward the negative half-plane. Then, for $k>0$, no eigenvalues of $M_{k}(E_{*})$ become positive, and diffusion cannot destabilize the equilibrium. Therefore, no Turing instability is induced by diffusion \cite{alma991008704579706535}.
\end{remark}
\section{PINNs Approximation}
\label{sec:pinn_method}
Solving system~\eqref{couple} numerically poses specific challenges that motivate the use of PINNs.
Classical mesh-based methods such as finite elements or finite differences require a spatial mesh whose resolution must be adapted to the three orders-of-magnitude gap between the fastest diffusion coefficient, and must handle the stiff vascular exchange term through costly implicit time-stepping. Beyond these technical difficulties, the five nonlinearly coupled equations require either operator-splitting  which introduces splitting errors or the assembly of large block Jacobians at each Newton step.

PINNs~\cite{raissi2019physics} circumvent all three issues simultaneously: spatial and temporal coordinates are treated as continuous inputs, derivatives are computed exactly by automatic differentiation with no truncation error, and all five equations are enforced simultaneously as soft constraints on a single trainable loss, with no mesh and no time-stepping stability condition to satisfy. The mathematical analysis of Section~\ref{sec:math} guarantees the existence and uniqueness of a solution
$\mathbf{U}=(T,B,O,I,S)^{\top}\in\bigl(C([0,\tau];L^{2}(\Omega))\cap
L^{2}(0,\tau;H^{1}(\Omega))\bigr)^{5}$
to system~\eqref{couple}.
To compute this solution numerically, we employ a PINNs
~\cite{raissi2019physics,farkane2024enhancing,farkane2025adaptive},
an approach that embeds the governing equations, initial conditions,
and boundary conditions directly into the loss functional of a deep neural
network, thereby transforming the PDE problem into an unconstrained
optimisation problem over the network weights. In this context recently many paper have focused on the using the deep leanrning methods for solving the he

The fundamental idea is to seek an approximation
\begin{equation}
  \widehat{\mathbf{U}}_{\theta}:\mathcal{Q}\longrightarrow\mathbb{R}^{5},
  \label{eq:pinn_map}
\end{equation}
realised by a single feedforward network
$\mathcal{N}_{\theta}:\mathbb{R}^{3}\to\mathbb{R}^{5}$
with trainable parameters $\theta$, such that
$\widehat{\mathbf{U}}_{\theta}\approx\mathbf{U}$ in the sense of
minimising a physics-based loss defined below.
Because no labelled simulation data are required, the method is
\emph{fully unsupervised}: the only supervision comes from the physical
laws encoded in~\eqref{couple}.

\subsection{Network Architecture}
\label{sec:arch}

The neural network $\mathcal{N}_{\theta}$ maps the spatio-temporal coordinate
$(x,y,t)\in\Omega\times(0,\tau)$ to the five-dimensional state vector:
\begin{equation}
  \mathcal{N}_{\theta}(x,y,t)
  =\bigl(\hat{T},\hat{B},\hat{O},\hat{I},\hat{S}\bigr)(x,y,t).
  \label{eq:net_io}
\end{equation}
We employ a fully connected Multi-Layer Perceptron (MLP) architecture consisting of $L$
 layers. The architecture comprises an input layer, several hidden layers of uniform width, and an output layer corresponding to the state variables.
The hidden-layer activation is the hyperbolic tangent
$\sigma(z)=\tanh(z)$. This choice is motivated by its infinite differentiability, which ensures that the high-order partial derivatives required for the differential operators in system~\eqref{couple} can be computed accurately through automatic differentiation without interpolation errors.
Let $\mathbf{h}^{(0)}=(x,y,t)^{\top}$; the forward pass reads
\begin{equation}
  \mathbf{h}^{(\ell)}
  = \sigma\!\left(W^{(\ell)}\mathbf{h}^{(\ell-1)}+\mathbf{b}^{(\ell)}\right),
  \quad \ell=1,\ldots,L-1,
  \label{eq:forward}
\end{equation}
\begin{equation}
  \mathcal{N}_{\theta}(x,y,t)
  = W^{(L)}\mathbf{h}^{(L-1)}+\mathbf{b}^{(L)},
  \label{eq:output}
\end{equation}
where $\{W^{(\ell)},\mathbf{b}^{(\ell)}\}_{\ell=1}^{L}$ are the weight
matrices and bias vectors collected in $\theta$.
Weights are initialised using the Xavier uniform scheme, which keeps the
variance of activations stable across layers and accelerates convergence
\cite{glorot2011deep}.

\subsection{Derivative Computation via Automatic Differentiation}
\label{sec:autodiff}

For every scalar network output $\hat{\varphi}\in\{\hat{T},\hat{B},\hat{O},\hat{I},\hat{S}\}$,
all partial derivatives required by the residual operators of~\eqref{couple} are
computed exactly by reverse-mode automatic differentiation
(AD)~\cite{paszke2019pytorch}, implemented in PyTorch.
Concretely:
\begin{itemize}
  \item \textbf{First-order derivatives} $\partial_{x}\hat{\varphi}$,
        $\partial_{y}\hat{\varphi}$, $\partial_{t}\hat{\varphi}$: one backward pass
        through the computation graph.
  \item \textbf{Laplacian} $\Delta\hat{\varphi}=\partial_{xx}\hat{\varphi}+\partial_{yy}\hat{\varphi}$:
        a second backward pass applied to $\partial_{x}\hat{\varphi}$ and
        $\partial_{y}\hat{\varphi}$, respectively.
\end{itemize}
This avoids any finite-difference truncation error and removes the mesh
dependency from derivative computation entirely.

\subsection{Collocation Point Sampling}
\label{sec:coll}

Three disjoint sets of collocation points are drawn to enforce
system~\eqref{couple}:

\medskip
\noindent\textbf{Interior (PDE) points} $\mathcal{S}_{\mathrm{pde}}$: A set of 
$N_{\mathrm{pde}}$ collocation  points
$(x_{k},y_{k},t_{k})\in Q$ sampled uniformly,
with $(x_{k},y_{k})\sim\mathbf{U}(\Omega)$ and
$t_{k}\sim\mathbf{U}(0,\tau)$.

\medskip
\noindent\textbf{Initial condition points} $\mathcal{S}_{\mathrm{ic}}$: A set of 
$N_{\mathrm{ic}}$ points $(x_{k},y_{k},0)\in\Omega\times\{0\}$
drawn uniformly over $\Omega$.

\medskip
\noindent\textbf{Boundary points} $\mathcal{S}_{\mathrm{bc}}$: A set of 
$N_{\mathrm{bc}}$ points $(x_{k},y_{k},t_{k})\in\partial\Omega\times(0,\tau)$,
drawn uniformly on each of the four sides of $\partial\Omega$,
with $t_{k}\sim\mathbf{U}(0,\tau)$.

\medskip
In practice  the union $\mathcal{S}=\mathcal{S}_{\mathrm{pde}}\cup\mathcal{S}_{\mathrm{ic}}
\cup\mathcal{S}_{\mathrm{bc}}$ is partitioned into
training (60\%), validation (20\%), and test (20\%) subsets.

\subsection{Composite Loss Functional}
\label{sec:loss}

The total loss to be minimised over $\theta$ is
\begin{equation}
  \mathcal{L}(\theta)
  =\lambda_{\mathrm{pde}}\,\mathcal{L}_{\mathrm{pde}}(\theta)
  +\lambda_{\mathrm{ic}}\,\mathcal{L}_{\mathrm{ic}}(\theta)
  +\lambda_{\mathrm{bc}}\,\mathcal{L}_{\mathrm{bc}}(\theta),
  \label{eq:total_loss}
\end{equation}
where the weights are set to
$\lambda_{\mathrm{pde}}$, $\lambda_{\mathrm{ic}}$,
$\lambda_{\mathrm{bc}}$.
These weights play an important role  in the convergence of the  PINNs loss. In particular, choosing $\lambda_{\mathrm{ic}}> \lambda_{\mathrm{pde}}$ ensures that the
network remains consistent with the prescribed initial conditions (see~\eqref{initial}). This weighting ensures that the learned spatio-temporal solution follows the given initial profiles. If the initial-condition term is assigned too little weight, the model may focus on reducing the PDE residual, which can cause it to converge to the trivial solution $\mathbf{U}\equiv 0$ instead of the intended physical solution.
\begin{enumerate}
    \item \textbf{PDE Residual Loss:} The five residual operators
$\mathcal{R}_{T},\mathcal{R}_{B},\mathcal{R}_{O},\mathcal{R}_{I},\mathcal{R}_{S}$
are obtained by substituting the network outputs into the left-hand sides
of~\eqref{couple}:
\begin{subequations}
\label{eq:residuals}
\begin{align}
\mathcal{R}_{T}
&=\partial_{t}\hat{T}-D_{T}\Delta\hat{T}
  -\rho_{T}\hat{T}\!\left(1-\tfrac{\hat{T}}{c}\right)
  \nonumber\\&+\delta_{T}\hat{T}+\alpha_{S}\hat{S}\hat{T};
\label{eq:resT}\\
\mathcal{R}_{B}
&=\partial_{t}\hat{B}-D_{B}\Delta\hat{B}
  -\rho_{B}\hat{B}\,\frac{K_{H}}{K_{H}+\hat{O}}
  \nonumber\\&+\delta_{B}\hat{B}+\beta_{I}\hat{I}\hat{B};
\label{eq:resB}\\
\mathcal{R}_{O}
&=\partial_{t}\hat{O}-D_{O}\Delta\hat{O}
  +\gamma_{T}\hat{T}\hat{O}\nonumber\\&-\gamma_{vas}(O_{\mathrm{vas}}-\hat{O});
\label{eq:resO}\\
\mathcal{R}_{I}
&=\partial_{t}\hat{I}-D_{I}\Delta\hat{I}
  -\beta_{T}\hat{T}+\delta_{I}\hat{I};
\label{eq:resI}\\
\mathcal{R}_{S}
&=\partial_{t}\hat{S}-D_{S}\Delta\hat{S}
  -\beta_{B}^{(S)}\hat{B}+\delta_{S}\hat{S}.
\label{eq:resS}
\end{align}
\end{subequations}
Setting $\mathcal{V}:=\{T,B,O,I,S\}$, the PDE loss is then the mean-squared residual over $\mathcal{S}_{\mathrm{pde}}$:
\begin{equation}
  \mathcal{L}_{\mathrm{pde}}(\theta)
  =\frac{1}{N_{\mathrm{pde}}}
   \sum_{k=1}^{N_{\mathrm{pde}}}
   \sum_{\varphi\in\mathcal{V}}
   \bigl[\mathcal{R}_{\varphi}(x_{k},y_{k},t_{k};\theta)\bigr]^{2}.
  \label{eq:Lpde}
\end{equation}
\item \textbf{Initial Condition Loss:} The initial condition loss enforces~\eqref{initial} at $t=0$:
\begin{equation}
  \mathcal{L}_{\mathrm{ic}}(\theta)
  =\frac{1}{N_{\mathrm{ic}}}
   \sum_{k=1}^{N_{\mathrm{ic}}}
   \sum_{\varphi\in\mathcal{V}}
   \bigl[\hat{\varphi}(x_{k},y_{k},0;\theta)-\varphi_{0}(x_{k},y_{k})\bigr]^{2},
  \label{eq:Lic}
\end{equation}
where $\varphi_{0}$ denotes the initial profile of each species
given in~\eqref{initial}.
Concretely, we consider the following initial conditions:
\begin{subequations}
\label{eq:ic_profiles}
\begin{align}
T_{0}(x,y)
&= \exp\!\left(-\frac{\|(x,y)-\mathbf{x}_{T}\|^{2}}
                                    {2\sigma_{T}^{2}}\right),
\label{eq:ic_T}\\
B_{0}(x,y)
&= 0.3\exp\!\left(-\frac{\|(x,y)-\mathbf{x}_{B}\|^{2}}
                                        {2\sigma_{B}^{2}}\right),
\label{eq:ic_B}\\
O_{0}(x,y)&=O_{\mathrm{vas}},\quad
I_{0}(x,y)=0.01,\,\nonumber\\
S_{0}(x,y)&=0,
\label{eq:ic_OIS}
\end{align}
\end{subequations}
where, the tumor centre is $\mathbf{x}_{T}=(L/2,L/2)$ ,
$\mathbf{x}_{B}=(0.3L,\,0.6L)$ is the bacterial injection site,
$\sigma_{T}=L/5$ and $\sigma_{B}=L/8$ are the spatial spreads, and $L=6$~mm.
The Gaussian profile~\eqref{eq:ic_T} initialises a localised tumour seed
at $10\%$ of the carrying capacity $\theta$, while~\eqref{eq:ic_B} places a
very small bacterial inoculum ($B_{0}\ll T_{0}$) at an off-centre position,
simulating a targeted injection away from the tumour core.
\item \textbf{Boundary Condition Loss}

The homogeneous Neumann conditions~\eqref{bord} are enforced by penalising the
squared gradient norm on $\mathcal{S}_{\mathrm{bc}}$:
\begin{equation}
  \mathcal{L}_{\mathrm{bc}}(\theta)
  =\frac{1}{N_{\mathrm{bc}}}
   \sum_{k=1}^{N_{\mathrm{bc}}}
   \sum_{\varphi\in\mathcal{V}}
\left\|\nabla_{(x,y)}\hat{\varphi}(x_{k},y_{k},t_{k};\theta)\right\|^{2}.
  \label{eq:Lbc}
\end{equation}
Note that the gradient $\nabla_{(x,y)}\hat{\varphi}$ is again computed
exactly via AD, so no approximation of the outward normal $\mathbf{n}$ is
needed.
\begin{proposition}[Consistency of the PINN loss]
\label{prop:consistency}
Let $\theta^{*}$ be such that
$\mathcal{L}(\theta^{*})=0$.
Then the network $\mathcal{N}_{\theta^{*}}$ satisfies the residuals
$\mathcal{R}_{\varphi}\equiv0$ for all $\varphi$ at every collocation
point in $\mathcal{S}_{\mathrm{pde}}$, reproduces the initial data at
every point in $\mathcal{S}_{\mathrm{ic}}$, and satisfies the zero-flux
condition at every point in $\mathcal{S}_{\mathrm{bc}}$.
\end{proposition}
\end{enumerate}
\begin{algorithm}[!t]
\caption{PINN Training for System~\eqref{couple}}
\label{alg:training}
\begin{algorithmic}[1]
\Require Collocation sets $\mathcal{S}_{\mathrm{pde}},
         \mathcal{S}_{\mathrm{ic}}, \mathcal{S}_{\mathrm{bc}}$;
         weights $\lambda_{\mathrm{pde}},\lambda_{\mathrm{ic}},
         \lambda_{\mathrm{bc}}$; A
         patience $P$; the maximum epochs $E$.
\Ensure Best-validation-loss weights $\theta^{*}$.
\State Initialise $\theta$ with Xavier uniform distribution.
\State $\mathcal{L}^{*}\leftarrow+\infty$;\; $p\leftarrow0$.
\For{$e=1,\ldots,E$}
  \State Sample mini-batch of size $512$ from each collocation set.
  \State Compute $\mathcal{L}(\theta)$ via forward pass + AD (Sec.~\ref{sec:autodiff}).
  \State Update $\theta\leftarrow\theta-\eta\,\nabla_{\theta}\mathcal{L}(\theta)$
         via Adam.
  \State Evaluate $\mathcal{L}_{\mathrm{val}}$ on the held-out
         validation set.
  \If{$\mathcal{L}_{\mathrm{val}}<\mathcal{L}^{*}$}
    \State $\mathcal{L}^{*}\leftarrow\mathcal{L}_{\mathrm{val}}$;\;
           $\theta^{*}\leftarrow\theta$;\; $p\leftarrow0$.
  \Else
    \State $p\leftarrow p+1$.
  \EndIf
  \If{$p\geq P$}  \textbf{break}  \Comment{Early stopping}
  \EndIf
  \State Adjust $\eta$ via ReduceLROnPlateau scheduler.
\EndFor
\State \Return $\theta^{*}$.
\end{algorithmic}
\end{algorithm}
\begin{table*}[!t]
\caption{The numerical simulation parameters of the system~\eqref{couple}}
\label{tab:params}
\renewcommand{\arraystretch}{1.25}
\centering
\begin{tabular}{@{}llccc@{}}
\toprule
\textbf{Equation} & \textbf{Description} & \textbf{Symbol} & \textbf{Value} & \textbf{Unit} \\
\midrule
\multirow{5}{*}{All}
  & Tumor cell motility         & $D_{T}$             & $0.01$   & mm$^{2}$\,day$^{-1}$ \\
  & Bacterial diffusivity       & $D_{B}$             & $0.1$    & mm$^{2}$\,day$^{-1}$ \\
  & Oxygen diffusivity          & $D_{O}$             & $1.0$    & mm$^{2}$\,day$^{-1}$ \\
  & Cytokine diffusivity        & $D_{I}$             & $0.5$    & mm$^{2}$\,day$^{-1}$ \\
  & AHL signal diffusivity      & $D_{S}$             & $0.3$    & mm$^{2}$\,day$^{-1}$ \\
\midrule
\multirow{4}{*}{(1) Tumor $T$}
  & Proliferation rate          & $\rho_{T}$          & $0.3$    & day$^{-1}$            \\
  & Carrying capacity           & $\theta$            & $1.0$    & normalised            \\
  & Apoptosis rate              & $\delta_{T}$        & $0.05$   & day$^{-1}$            \\
  & Signal suppression coeff.   & $\alpha_{S}$        & $0.2$    & day$^{-1}$    \\
\midrule
\multirow{4}{*}{(2) Bacteria $B$}
  & Growth/clearance rate       & $\rho_{B}$          & $0.5$    & day$^{-1}$                          \\
  & O$_{2}$ half-saturation     & $K_{H}$             & $0.1$    & normalised                          \\
  & Natural death rate          & $\delta_{B}$        & $0.1$    & day$^{-1}$                          \\
  & Immune clearance efficiency & $\beta_{I}$         & $0.3$    & day$^{-1}$        \\
\midrule
\multirow{3}{*}{(3) Oxygen $O$}
  & Tumor O$_{2}$ consumption   & $\gamma_{T}$        & $0.2$    & day$^{-1}$  \\
  & Vascular exchange rate      & $\gamma_{E}$        & $0.5$    & day$^{-1}$  \\
  & Baseline O$_{2}$ level      & $O_{\mathrm{ext}}$  & $0.2$    & normalised  \\
\midrule
\multirow{2}{*}{(4) Cytokines $I$}
  & Tumor-derived production    & $\beta_{T}$         & $0.1$    & day$^{-1}$ \\
  & Cytokine decay rate         & $\delta_{I}$        & $0.2$    & day$^{-1}$                               \\
\midrule
\multirow{2}{*}{(5) Signal $S$}
  &  production rate         & $\beta_{B}$   & $0.4$    & day$^{-1}$ \\
  &  degradation rate        & $\delta_{S}$        & $0.3$    & day$^{-1}$           \\
\bottomrule
\end{tabular}
\end{table*}
More details about the numerical parameters and implementation details are described in the section \ref{sec:num}.\\
In the next section, we provide a rigorous convergence analysis of the proposed PINN framework and establish error estimates for the method. 
\section{Convergence Analysis of the PINN Approximation}\label{sec:convergence}
Let $\mathbf{U}=(T,B,O,I,S)^{\top}$ be the unique classical solution of system~\eqref{couple} corresponding to the initial data $\mathbf{U}_{0}$ (guaranteed by Theorem~\ref{wp}-Proposition \ref{propReg}), and let
$\widehat{\mathbf{U}}_{\theta} = (\hat{T},\hat{B},\hat{O},\hat{I},\hat{S})^{\top}$
be the PINN approximation produced by the network
$\mathcal{N}_{\theta}:\mathbb{R}^{3}\to\mathbb{R}^{5}$ after minimizing the composite loss~\eqref{eq:total_loss}.
Define the \emph{total approximation error} as
\begin{equation}
  \mathbf{e}(t,\mathbf{x})
  \;=\; \mathbf{U}(t,\mathbf{x}) - \widehat{\mathbf{U}}_{\theta}(t,\mathbf{x}),
  \quad (t,\mathbf{x})\in Q.
  \label{eq:error_def}
\end{equation}
The total error admits the three-way decomposition:
\begin{gather}
\begin{split}
  \|\mathbf{e}\|_{L^{\infty}(0,\tau;(L^{2}(\Omega))^{5})}
  &\leq \!\!\!\!
  \underbrace{\mathcal{E}_{\mathrm{approx}}(\theta)}_{\text{(I) approximation}}\!
  +\!
  \underbrace{\mathcal{E}_{\mathrm{gen}}(N)}_{\text{(II) generalization}}\\
 & +
  \underbrace{\mathcal{E}_{\mathrm{opt}}(\theta^{*})}_{\text{(III) optimisation}},
  \label{eq:three_way}
  \end{split}
\end{gather}
where:
\begin{itemize}
  \item[(I)] \textbf{Approximation error}: the best-achievable error of
        a depth-$L$, width-$n$ network in approximating the true solution
        $\mathbf{U}\in (L^{2}(0,\tau;H^{1}(\Omega)))^{5}$.
  \item[(II)] \textbf{Generalization error}: the discrepancy between the
        empirical (Monte Carlo) loss evaluated on $N_{\mathrm{pde}}$,
        $N_{\mathrm{ic}}$, $N_{\mathrm{bc}}$ collocation points and the
        continuous $L^{2}$-residuals over $Q$, $\Omega$, and $\Sigma$.
  \item[(III)] \textbf{Optimisation error}: the gap between the loss
        achieved by the Adam algorithm after $E$ epochs and the global
        minimum of the loss landscape.
\end{itemize}
Item~(III) depends on the specific optimization trajectory and is not addressed here; we provide upper bounds for~(I) and~(II), and we derive a conditional stability estimate showing that if~(I) and~(II) are small, then $\|\mathbf{e}\|_{L^{\infty}(0,\tau;(L^{2}(\Omega))^{5})}$ is small.

We begin by defining the continuous residual functionals associated with the PINN approximation $\mathbf{U}_{\theta}$.
\subsection{Continuous Residual Functionals}
\label{sec:residuals_cont}

Define the \emph{continuous} (true) residuals associated with the
PINN approximation $\widehat{\mathbf{U}}_{\theta}$:
\begin{align}
  \mathcal{E}_{\mathrm{pde}}(\theta)
  &\;=\;
  \int_{0}^{\tau}\!\!\int_{\Omega}
  \sum_{\varphi\in\mathcal{V}}
  \bigl|\mathcal{R}_{\varphi}(\widehat{\mathbf{U}}_{\theta};t,\mathbf{x})\bigr|^{2}
  \,\mathrm{d}\mathbf{x}\,\mathrm{d}t,
  \label{eq:eps_pde}\\[4pt]
  \mathcal{E}_{\mathrm{ic}}(\theta)
  &\;=\;
  \int_{\Omega}
  \sum_{\varphi\in\mathcal{V}}
  \bigl|\hat{\varphi}(\mathbf{x},0;\theta)-\varphi_{0}(\mathbf{x})\bigr|^{2}
  \,\mathrm{d}\mathbf{x},
  \label{eq:eps_ic}\\[4pt]
  \mathcal{E}_{\mathrm{bc}}(\theta)
  &\;=\;
  \int_{0}^{\tau}\!\!\int_{\partial\Omega}
  \sum_{\varphi\in\mathcal{V}}
  \bigl|\nabla_{\mathbf{x}}\hat{\varphi}(\mathbf{x},t;\theta)\cdot\mathbf{n}\bigr|^{2}
  \,\mathrm{d}\sigma\,\mathrm{d}t,
  \label{eq:eps_bc}
\end{align}
where 
$\mathcal{R}_{\varphi}$ are the five residual operators defined
in~\eqref{eq:residuals}.
Let
$\mathcal{E}(\theta)
= \lambda_{\mathrm{pde}}\,\mathcal{E}_{\mathrm{pde}}
+\lambda_{\mathrm{ic}}\,\mathcal{E}_{\mathrm{ic}}
+\lambda_{\mathrm{bc}}\,\mathcal{E}_{\mathrm{bc}}$
denote the total continuous loss.

\subsection{Lipschitz Continuity of the Nonlinear Operator}
\label{sec:lipschitz}
In this section, we first establish that the nonlinear function 
$\mF:\mathbb{R}^{5}\to\mathbb{R}^{5}$ defined in~\eqref{F} is locally Lipschitz.
\begin{lemma}[Lipschitz bound on $\mF$]
\label{lem:lipschitz}
There exists a constant $L_{F}>0$ such that
\begin{equation}
  \|\mF(\mathbf{U})-\mF(\widehat{\mathbf{U}}_{\theta})\|_{(L^{2}(\Omega))^{5}}
  \;\leq\; L_{\mF}\,\|\mathbf{e}\|_{(L^{2}(\Omega))^{5}},
  \quad \text{a.e.}\ t\in(0,\tau).
  \label{eq:lip}
\end{equation}
\end{lemma}
\begin{proof}

Theorem~\ref{wp} guarantees $\mathbf{U}\in(L^{\infty}(Q)_{+})^{5}$ whenever
$\mathbf{U}_{0}\in(H^{1} (\Omega) \cap (L^{\infty}(\Omega)_{+})^{5}$.
For the network approximation $\widehat{\mathbf{U}}_{\theta}$, boundedness is enforced by the bounded weight initialisation and the $\tanh$ activation (whose output is always in $(-1,1)$). Then, there exists constants $M$ such that $\|\mathbf{U}\|_{(L^{\infty}(Q))^{5}}\leq M$, and  $\|\widehat{\mathbf{U}}_{\theta}\|_{(L^{\infty}(Q))^{5}}\leq M$.
Now, recall from~\eqref{F} the five components of  $\mF$:
\begin{align*}
\begin{split}
F_{T}(\mathbf{U})&=\rho_{T }T\left(1 - \frac{T}{c}\right)-\delta_{T} T-\alpha_{S} S T;\\
  F_{B}(\mathbf{U})&=\rho_{B}B\frac{K_{H}}{K_{H}+O}-\delta_{B}B-\beta_{I}IB;\\
  F_{O}(\mathbf{U})&=-\gamma_{T} T O +\gamma_{vas} (O_{vas}-O); \\
  F_{I}(\mathbf{U})&= \beta_{T}T-\delta_{I}I;\\
 F_{S}(\mathbf{U})&= \beta_{B}B-\delta_{S}S.
 \end{split}
   \end{align*}
We bound $F_{\varphi}(\mathbf{U})-F_{\varphi}(\widehat{\mathbf{U}}_{\theta})$
in $L^{2}(\Omega)$ with $\varphi=T,B,O,I,S$.
Let $e_{\varphi}=\varphi-\hat{\varphi}$ for each component.\\
\noindent\textbf{Term $F_{T}$:}
\begin{gather*}
\begin{split}
 \hspace{-0.9cm} F_{T}(\mathbf{U}) - F_{T}(\widehat{\mathbf{U}}_{\theta})
  &= (\rho_{T}-\delta_{T})e_{T}-\frac{\rho_{T}}{c}(T+\hat{T})e_{T}\\
  &-\alpha_{S}(ST-\hat{S}\hat{T}).
  \end{split}
\end{gather*}
Since $|T|\leq M,$ $|\hat{T}|\leq M$ and using the assumption \textit{\textbf{(H)}}, the first term is bounded by $C_{M}\|e_{T}\|_{L^{2}(\Omega)}$.
For the second term, using the fact that
$ST-\hat{S}\hat{T}= S\,e_{T}+\hat{T}\,e_{S}$, and  $|S|\leq M$, $|\hat{T}|\leq {M}$:
\begin{equation*}
  \|ST-\hat{S}\hat{T}\|_{L^{2}(\Omega)} \leq C_{M}\|e_{T}\|_{L^{2}(\Omega)}+C_{M}\|e_{S}\|_{L^{2}(\Omega)}.
\end{equation*}
Hence
\begin{gather}\label{FT}
\hspace{-0.9cm}\|F_{T}(\mathbf{U})-F_{T}(\widehat{\mathbf{U}}_{\theta})\|_{L^{2}(\Omega)}\leq C_{M}\big(\|e_{T}\|_{L^{2}(\Omega)}+\|e_{S}\|_{L^{2}(\Omega)}\big). 
\end{gather}

\medskip
\noindent\textbf{Term $F_{B}$:}
For the fractional term, we have
\begin{equation*}
  B\frac{ K_{H}}{K_{H}+O}-\hat{B}\frac{K_{H}}{K_{H}+\hat{O}}
  = \frac{(K^{2}_{H}+K_{H})\hat{O}e_{B}-K_{H}\hat{B}e_{O}}{(K_{H}+O)(K_{H}+\hat{O})}.
\end{equation*}
Since $O,\hat{O}\geq0$, we have $(K_{H}+O)(K_{H}+\hat{O})\ge K^{2}_{H}$, and using
 $|\hat{B}|\leq M$, $|\hat{O}|\leq M$ and assumption \textit{(\textbf{H)}}, we obtain:
\begin{equation*}
  \left\|B\frac{K_{H}}{K_{H}+O}-\hat{B}\frac{\hat{B}}{K_{H}+\hat{O}}\right\|_{L^{2}(\Omega)}
  \!\!\!\!\!\leq C_{M}(\|e_{B}\|_{L^{2}(\Omega)} +\|e_{O}\|_{L^{2}(\Omega)}).
\end{equation*}
For the rest of terms, we have $$-\delta_{B}B+\delta_{B}\hat{B}-\beta_{I}IB+\beta_{I}\hat{I}\hat{B}=-\delta_{B}e_{B}-\beta_{I}( I\,e_{B}+\hat{B}\,e_{I}).$$  Using
 $|\hat{I}|\leq M$, $|\hat{B}|\leq M$ and assumption \textit{(\textbf{H)}}, we obtain:
 \begin{gather}\label{FB}
 \begin{split}
 \|F_{B}(\mathbf{U})-F_{B}(\widehat{\mathbf{U}}_{\theta})&\|_{L^{2}(\Omega)}
\leq C_{M}\big(\|e_{B}\|_{L^{2}(\Omega)}\\
&+\|e_{O}\|_{L^{2}(\Omega)}+\|e_{I}\|_{L^{2}(\Omega)}\big). 
\end{split}
 \end{gather}
\textbf{Terms $F_{O}$, $F_{I}$, $F_{S}$:}
Using $|T|,|B|,|O|\leq M$, we have similarly obtained as the previous estimates, the following:
\begin{gather}\label{FO}
 \begin{split}
\|F_{O}(\mathbf{U})-F_{O}(\widehat{\mathbf{U}}_{\theta})&\|_{L^{2}(\Omega)}
\leq C_{M}\big(\|e_{T}\|_{L^{2}(\Omega)}\\
&+\|e_{B}\|_{L^{2}(\Omega)}
+\|e_{O}\|_{L^{2}(\Omega)}\big);
 \end{split}
 \end{gather}
 \begin{gather}\label{FI}
\hspace{-1cm}\|F_{I}(\mathbf{U})-F_{I}(\widehat{\mathbf{U}}_{\theta})\|_{L^{2}(\Omega)}\leq C_{M}\big(\|e_{T}\|_{L^{2}(\Omega)}+\|e_{I}\|_{L^{2}(\Omega)}\big);
 \end{gather}
 \begin{gather}\label{FS}
\hspace{-1cm}\|F_{S}(\mathbf{U})-F_{S}(\widehat{\mathbf{U}}_{\theta})\|_{L^{2}(\Omega)}
\leq C_{M}\big(\|e_{B}\|_{L^{2}(\Omega)}+\|e_{S}\|_{L^{2}(\Omega)}\big).
 \end{gather}
Summing estimates \eqref{FT}-\eqref{FS} with $\mathbf{e}=(e_{T},e_{B},e_{O},e_{I},e_{S})^{\top}$, and setting $L_{F}=C_{M}$, completes the proof of Lemma \ref{lem:lipschitz}.
\end{proof}

\subsubsection{Conditional Stability: Error Controlled by Residuals}
\label{sec:stability}

The central result of this section shows that small PDE, IC, and BC residuals
imply a small approximation error.

\begin{theorem}[Conditional Stability]
\label{thm:stability}
Let the assumptions of Theorem~\ref{wp} hold and let
$\widehat{\mathbf{U}}_{\theta}\in (L^{2}(0,\tau;H^{1}(\Omega))\cap
H^{1}(0,\tau;(H^{1}(\Omega))'))^{5}$
be any network approximation.
Define the total residual
$\mathcal{E}(\theta)
=\mathcal{E}_{\mathrm{pde}}(\theta)
+\mathcal{E}_{\mathrm{ic}}(\theta)
+\mathcal{E}_{\mathrm{bc}}(\theta)$.
Then there exists a constant $\mathcal{C}=\mathcal{C}(\tau,M,L_{F},d_{0})>0$,
independent of $\theta$, such that
\begin{equation}
  \sup_{t\in[0,\tau]}
  \|\mathbf{e}(t)\|^{2}_{(L^{2}(\Omega))^{5}}
  +
  d_{0}
  \int_{0}^{\tau}
  \|\nabla\mathbf{e}(s)\|^{2}_{(L^{2}(\Omega))^{5}}
  \,\mathrm{d}s
  \;\leq\;
\mathcal{C}\,\mathcal{E}(\theta).
  \label{eq:stability}
\end{equation}
\end{theorem}

\begin{proof}
\textbf{Step 1 — Error equation.}
Since $\mathbf{U}$ satisfies~(10) exactly and
$\widehat{\mathbf{U}}_{\theta}$ satisfies it with residual
\\
$\mathbf{R}_{\theta}=({\mathcal{R}_{T}},{\mathcal{R}_{B}},
{\mathcal{R}_{O}},{\mathcal{R}_{I}},{\mathcal{R}_{S}})^{\top}$,
subtracting gives for a.e.~$t\in(0,\tau)$:
\begin{subequations}
\label{eq:error_system}
\begin{align}
  \partial_{t}\mathbf{e}
  + \mathbf{A}\mathbf{e}
  &= \mF(\mathbf{U})-\mF(\widehat{\mathbf{U}}_{\theta})
    - \mathbf{R}_{\theta}
    \quad\text{in }Q,
  \label{eq:err_eq}\\
  \partial_{\mathbf{n}}\mathbf{e}
  &= -\partial_{\mathbf{n}}\widehat{\mathbf{U}}_{\theta}
    \quad\text{on }\Sigma,
  \label{eq:err_bc}\\
  \mathbf{e}(0,\cdot)
  &= \mathbf{U}_{0}-\widehat{\mathbf{U}}_{\theta}(\cdot,0)
    \quad\text{in }\Omega.
  \label{eq:err_ic}
\end{align}
\end{subequations}

\textbf{Step 2 — Energy identity.}
Take the $(L^{2}(\Omega))^{5}$ inner product of~\eqref{eq:err_eq}
with $\mathbf{e}$ and integrate over $\Omega$.
Integration by parts, using the boundary condition~\eqref{eq:err_bc}, yields
\begin{align}
\frac{1}{2}\frac{\mathrm{d}}{\mathrm{d}t}
&\|\mathbf{e}(t)\|^{2}_{(L^{2}(\Omega))^{5}}
+
\bigl(\mathbf{D}\nabla\mathbf{e},\nabla\mathbf{e}\bigr)_{(L^{2}(\Omega))^{5}}
\notag\\
&= \bigl(F(\mathbf{U})-F(\widehat{\mathbf{U}}_{\theta}),\mathbf{e}\bigr)_{(L^{2})^{5}}
   -\bigl(\mathbf{R}_{\theta},\mathbf{e}\bigr)_{(L^{2}(\Omega))^{5}}
\notag\\
&\quad
  +\int_{\partial\Omega}\mathbf{D}\,\partial_{\mathbf{n}}
  \widehat{\mathbf{U}}_{\theta}\cdot\mathbf{e}\,\mathrm{d}\sigma.
\label{eq:energy_id}
\end{align}

\textbf{Step 3 — Bounding the right-hand side of~\eqref{eq:energy_id}.}

\emph{Nonlinear term.}
By Lemma~\ref{lem:lipschitz}:
\begin{equation}
\bigl(F(\mathbf{U})-F(\widehat{\mathbf{U}}_{\theta}),\mathbf{e}\bigr)_{(L^{2}(\Omega))^{5}}
\leq L_{F}\|\mathbf{e}\|^{2}_{(L^{2}(\Omega))^{5}}.
\label{eq:nl_bound} 
\end{equation}

\emph{Residual term.}
By Young's inequality  for all $\varepsilon>0$:
\begin{equation}
\bigl|(\mathbf{R}_{\theta},\mathbf{e})_{(L^{2}(\Omega))^{5}}\bigr|
\leq
\frac{1}{2\varepsilon}\|\mathbf{R}_{\theta}\|^{2}_{(L^{2}(\Omega))^{5}}
+\frac{\varepsilon}{2}\|\mathbf{e}\|^{2}_{(L^{2}(\Omega))^{5}}.
\label{eq:res_bound}
\end{equation}

\emph{Boundary term.}
Using the trace inequality in $L^{2}(\partial\Omega)$ (see \cite[Theorem 1.6.6.]{brenner2008mathematical}), and Young's inequality, we have for all $\varepsilon>0$:
\begin{align}
\hspace{-1cm}\left|\int_{\partial\Omega}\mathbf{D}\,\partial_{\mathbf{n}}
\widehat{\mathbf{U}}_{\theta}\cdot\mathbf{e}\,\mathrm{d}\sigma\right|
&\leq
|||\mathbf{D}|||_{\infty}
\|\partial_{\mathbf{n}}\widehat{\mathbf{U}}_{\theta}\|_{(L^{2}(\partial\Omega))^{5}}
\|\mathbf{e}\|_{(L^{2}(\partial\Omega))^{5}}
\notag\\
&\leq
\frac{C_{d,\Omega}}{2\varepsilon}
\|\partial_{\mathbf{n}}\widehat{\mathbf{U}}_{\theta}\|^{2}_{(L^{2}(\partial\Omega))^{5}}
\nonumber\\&+\frac{\varepsilon C_{d,\Omega}}{2}
\bigl(\varepsilon\|\nabla\mathbf{e}\|^{2}_{(L^{2}(\Omega))^{5}}
+\varepsilon^{-1}\|\mathbf{e}\|^{2}_{(L^{2}(\Omega))^{5}}\bigr).
\label{eq:bc_bound}
\end{align}

\textbf{Step 4 — Differential inequality.}
Since from \eqref{Eqstate}, $\mathbf{A}=-\mathbf{D}\mathbf{\Delta}$ with $\mathbf{D}=\mathrm{diag}(D_{T},D_{B},D_{O},D_{I},D_{S})$,
we have 
$(\mathbf{D}\nabla\mathbf{e},\nabla\mathbf{e})_{(L^{2}(\Omega))^{5}}\geq d_{0}\|\nabla\mathbf{e}\|^{2}_{(L^{2}(\Omega))^{5}}$,
where $d_{0}=\min(D_{T},D_{B},D_{O},D_{I},D_{S})$.
Substituting bounds~\eqref{eq:nl_bound}--\eqref{eq:bc_bound} into~\eqref{eq:energy_id},
and choosing $\varepsilon=\sqrt{d_{0}/2C_{d,\Omega}}$ to absorb the gradient term:
\begin{gather}
\begin{split}
\frac{\mathrm{d}}{\mathrm{d}t}
\|\mathbf{e}\|^{2}_{(L^{2}(\Omega))^{5}}
+ d_{0}\|\nabla\mathbf{e}\|^{2}_{(L^{2})^{5}}
&\leq
\Lambda\|\mathbf{e}\|^{2}_{(L^{2}(\Omega))^{5}}
\\& + C_{1}\|\mathbf{R}_{\theta}\|^{2}_{(L^{2}(\Omega))^{5}}
\\&+C_{2}\|\partial_{\mathbf{n}}\widehat{\mathbf{U}}_{\theta}\|^{2}_{(L^{2}(\partial\Omega))^{5}},
\label{eq:gronwall_ineq}    
\end{split}
\end{gather}
where
$\Lambda = 2L_{F}+\varepsilon+C_{d,\Omega}$,
$C_{1}=\varepsilon^{-1}$,
$C_{2}=\varepsilon^{-1}C_{d,\Omega}$
are constants depending only on $M$, $d_{0}$, and the trace constant $C_{d,\Omega}$.

\textbf{Step 5 — Application of Grönwall’s Lemma.}
Integrating~\eqref{eq:gronwall_ineq} over $[0,t]$ and applying
Grönwall's Lemma:
\begin{align}
\|\mathbf{e}(t)\|^{2}_{(L^{2}(\Omega))^{5}}
&+d_{0}\!\int_{0}^{t}\!\|\nabla\mathbf{e}\|^{2}_{(L^{2}(\Omega))^{5}}\mathrm{d}s
\notag\nonumber\\
&\leq e^{\Lambda\tau}
\Biggl[
\|\mathbf{e}(0)\|^{2}_{(L^{2}(\Omega))^{5}}
\nonumber\\&+C_{1}\!\int_{0}^{\tau}\!\!\int_{\Omega}
\sum_{\varphi\in\mathcal{V}}|\mathcal{R}_{\varphi}|^{2}\mathrm{d}\mathbf{x}\,\mathrm{d}s
\notag \nonumber\\
&
+C_{2}\!\int_{0}^{\tau}\!\!\int_{\partial\Omega}
\sum_{\varphi\in\mathcal{V}}|\partial_{\mathbf{n}}\hat{\varphi}|^{2}\mathrm{d}\sigma\,\mathrm{d}s
\Biggr].
\label{eq:gronwall_out}
\end{align}
Recognising
\begin{gather*}
 \|\mathbf{e}(0)\|^{2}_{(L^{2})^{5}}=\mathcal{E}_{\mathrm{ic}}(\theta);\quad
 \int_{Q}\sum|\mathcal{R}_{\varphi}|^{2} =\mathcal{E}_{\mathrm{pde}}(\theta);\\
\int_{\Sigma}\sum|\partial_{\mathbf{n}}\hat{\varphi}|^{2}
=\mathcal{E}_{\mathrm{bc}}(\theta),
\end{gather*}
and taking the supremum over $t\in[0,\tau]$ gives~\eqref{eq:stability}
with
$\mathcal{C}=e^{\Lambda\tau}\max(1,C_{1}/\lambda_{\mathrm{ic}},
C_{2}/\lambda_{\mathrm{bc}},C_{1}/\lambda_{\mathrm{pde}})$.
\end{proof}

\begin{remark}[Role of the IC weight $\lambda_{\mathrm{ic}}$]
Estimate~\eqref{eq:stability} shows that the error bound is proportional
to $\mathcal{E}(\theta)$, which includes $\lambda_{\mathrm{ic}}\mathcal{E}_{\mathrm{ic}}$.
A large $\lambda_{\mathrm{ic}}$ forces the optimiser to reduce
$\mathcal{E}_{\mathrm{ic}}$ aggressively, thereby tightening the bound.
This provides  justification for the choice a higher
$\lambda_{\mathrm{ic}}$ used in training.
\end{remark}
\subsubsection{Neural Network Approximation Capacity}
\label{sec:approx}

We now address the question of a neural network's ability to approximate the exact solution of the pair system. The Proposition \ref{propReg} gives the Sobolev regularities of the solution $\bU$ of the system \eqref{couple} corresponding to the initial data $\mathbf{U}_0 \in (L^{\infty}(\Omega)_+ \cap \Hp{1}(\Omega))^5$. Precisely, we have $\bU \in \bigl(\Lp{2}(0,\tau;\,\Hp{2}(\Omega)) \cap \Hp{1}(0,\tau;\,\Lp{2}(\Omega))\bigr)^5
=: \bigl(\Hp{2,1}(Q)\bigr)^5.$\\
Based on well-known classical results in the analysis of parabolic equations, this is natural Sobolev regularity of solutions to second-order parabolic systems with initial data $H^{1}(\Omega)$. 

The regularity $\Hp{2,1}(\mathcal{Q})$ of the global solution of \eqref{couple} is sufficient for estimating conditional stability (Theorem~\ref{thm:stability}). However, the approximation theorem (Theorem~\ref{thm:approx}) requires additional regularity of the global solution (which we do not have with the initial data in $(H^{1}(\Omega))^{5}$ to approximate residual operators related to derivatives within second-order spaces. At least one additional order of differentiability beyond the largest derivative order $K$ is required to approximate ($K=2$ for \eqref{couple}), according to the principle $m \geq K+1$ using in \cite[Theorem 1.] {girault2024approximation} and ~\cite[Theorem 5.1]{de2021approximation},  where $m$ is the maximum order of derivatives used in the Sobolev norm. To ensure this additional regularity of the global solution, it is assumed in only this section that the initial data $\bU_0 \in (H^{2}(\Omega))^5$. Using the same classical techniques in the proof of Proposition \ref{propReg}, we have that the associated global solution $\bU$ of \eqref{couple} satisfies
\begin{equation*}
\bU \in \bigl(\Lp{2}(0,\tau;\,\Hp{3}(\Omega))
\;\cap\; \Hp{2}(0,\tau;\,\Lp{2}(\Omega))\bigr)^5.
\end{equation*}
\begin{theorem}[Approximation by $\tanh$ networks]
\label{thm:approx}
Let the solution~$\bU$ of system~\ref{couple} corresponding to the initial data $\bU_{0}\in(H^{2}(\Omega))^{5}$ satisfying
\[
\bU \in \bigl(\Lp{2}(0,\tau;\,\Hp{3}(\Omega))
\;\cap\; \Hp{2}(0,\tau;\,\Lp{2}(\Omega))\bigr)^5=:\cX.
\]
Then for any $\varepsilon > 0$, there exists a fully connected $\tanh$ network $\bN_\theta : \mathbb{R}^3 \to \mathbb{R}^5$ with~$L$ layers of width~$n$ satisfying
\begin{equation}\label{eq:approx_qual}
\min_{\theta}\, \cL(\theta) \leq \varepsilon,
\end{equation}
provided that $n$ is chosen large enough. More precisely, for the continuous loss $\cR_{\mathrm{pde}} + \cR_{\mathrm{ic}} + \cR_{\mathrm{bc}}$\emph{:}
\begin{equation}\label{eq:approx_rate}
\min_{\theta}\, \bigl(\cR_{\mathrm{pde}}(\theta)
+ \cR_{\mathrm{ic}}(\theta) + \cR_{\mathrm{bc}}(\theta)\bigr)
\;\leq\;
C_{\mathrm{approx}}\;\frac{\ln^4(n)}{n^2}\;
\|\bU\|_{\cX}^2,
\end{equation}
where
\begin{equation*}\label{eq:Xnorm}
\|\bU\|_{\cX}^2
:= \lVert\bU\rVert^2_{(\Lp{2}(0,\tau;\,\Hp{3}(\Omega)))^5}
+ \lVert\bU\rVert^2_{(\Hp{2}(0,\tau;\,\Lp{2}(\Omega)))^5},
\end{equation*}
and $C_{\mathrm{approx}}$ depends on $\tau$, $|\Omega|$, and the system parameters $(D_T, D_B, D_O, D_I, D_S, L_F)$ but is independent of $\theta$ and~$n$.
\end{theorem}

\begin{proof}
The proof combines the Sobolev approximation theory for $\tanh$ networks established in~\cite[Theorem~1 and Corollary~2]{girault2024approximation} with the structure of the residual operators~\ref{eq:resT}-\ref{eq:resS}.

\smallskip
\noindent\textit{Step~1 - Component-wise approximation.}
As explained at the beginning of the subsection, the residual operators $\cR_\varphi$ defined in~\ref{eq:resT}-\ref{eq:resS} use at most $K=2$ derivatives (the Laplacian $\Delta$ represented by second-order spatial derivatives and $\partial_t$ represented by a first-order temporal derivative). Therefore, to bound $\cR_{\mathrm{pde}}$ in $\Lp{2}(Q)$, we must control the approximation error in the Sobolev norm $\Wp{K}{2}(Q)$. To do this, we apply a tighter approximation bound for Sobolev-regular functions (see~\cite[Theorem~1]{girault2024approximation} with $m=K+1=3$) component-wise to each variable $\varphi \in \{T, B, O, I, S\}$. Then, there exists a $\tanh$ network $\hat{\varphi}_n$ with width parameter~$n$ such that for all $0 \leq k \leq 2$:
\begin{equation}\label{eq:comp_approx}
\hspace{-0.17cm}\|\varphi - \hat{\varphi}_n\|_{\Wp{k}{2}(Q)}
\;\leq\;
C_{k,3,3}^{[\mathrm{approx}]}\;\Delta_{k,3,3}^{[\tanh]}(\delta;\alpha_{\epsilon,\delta})\;
|\varphi|_{\Wp{3}{2}(Q)},
\end{equation}
where the asymptotic equivalence term
\begin{gather}\label{eq:Delta_asymp}
\Delta_{k,3,3}^{[\tanh]}(\delta;\alpha)
\;\sim\;
\left(\frac{m + d + K}{2}\right)^k
\ln^k\!\left(\frac{1}{\delta}\right)\,\delta^{m-k},
\end{gather}
with $m=3$, $d= 3$, $K=2$, and $n \sim \delta^{-1}$. In our situation, the dominant asymptotic equivalence term is obtained for $(k=K=2)$. Then the terms $C_{2,3,3}^{[\mathrm{approx}]}$,$\Delta_{2,3,3}^{[\tanh]}(\delta)$ are given by
\begin{gather}\label{asym}
\begin{split}
  \Delta_{2,3,3}^{[\tanh]}(\delta)\;\sim\; 16\;\ln^2\!\left(\frac{1}{\delta}\right)\,\delta\;=\; 16\;\frac{\ln^2(n)}{n},\\ C^{[\mathrm{approx}]}_{2,3,3}\sim2\times3^{9/2}(1+\zeta)\sqrt{6}\left(\frac{3\sqrt{3}}{\pi}\right)^{3}.
  \end{split}
\end{gather}
\noindent\textit{Step~2 --- PDE residual bound.}
Since $\bU$ satisfies~\eqref{couple}, the residual for each component takes the form
\begin{gather*}
    \cR_\varphi(\bN_\theta)
= \partial_t e_\varphi - D_\varphi\,\Delta e_\varphi
- \bigl[F_\varphi(\bU) - F_\varphi(\bUhat_\theta)\bigr],
\end{gather*}
where $e_\varphi = \varphi - \hat{\varphi}_n$. 
Then we have:
\begin{gather*}
\begin{split}
 \|\cR_\varphi\|_{\Lp{2}(Q)}
\;\leq\;
\|\partial_t e_\varphi&\|_{\Lp{2}(Q)}
+ D_\varphi\,\|\Delta e_\varphi\|_{\Lp{2}(Q)}\\
&+ \|F_\varphi(\bU) - F_\varphi(\bUhat_\theta)\|_{\Lp{2}(Q)}. 
\end{split}
\end{gather*}
The first two terms are controlled by the $\Wp{1}{2}(Q)$ and $\Wp{2}{2}(Q)$ norm of approximation errors from Step~1. The nonlinear term is bounded by the Lipschitz estimate of Lemma~\ref{lem:lipschitz}: 
\begin{gather*}
\|F_\varphi(\bU) - F_\varphi(\bUhat_\theta)\|_{\Lp{2}(Q)}
\;\leq\; L_F\,\|\be\|_{(\Lp{2}(Q))^5}.
\end{gather*}
Therefore:
\begin{gather*}
\begin{split}
\cR_{\mathrm{pde}}(\theta)
&= \int_Q \sum_{\varphi \in \mathcal{S}} |\cR_\varphi|^2 \dx\,\dt\\
&\;\leq\;
C\sum_{\varphi \in \mathcal{S}} \sum_{k=0}^{2}
\|\varphi - \hat{\varphi}_n\|^2_{\Wp{k}{2}(Q)}.
\end{split}
\end{gather*}
Applying~\eqref{eq:comp_approx} for $k=0,1,2$ and using~\eqref{eq:Delta_asymp}, we obtain
\begin{equation}\label{eq:Rpde_bound}
\cR_{\mathrm{pde}}(\theta)
\;\leq\;
C\;\frac{\ln^4(n)}{n^2}\;
\sum_{\varphi \in \mathcal{S}} \lVert\varphi\lVert^2_{\Wp{3}{2}(Q)}.
\end{equation}

\smallskip
\noindent\textit{Step~3 --- IC and BC terms.} By the trace Theorem, the initial condition error satisfies:
\begin{gather*}
    \begin{split}
\cR_{\mathrm{ic}}(\theta)
&\leq
C_{\mathrm{}} \sum_{\varphi \in \mathcal{S}}
\|\varphi - \hat{\varphi}_n\|^2_{\Wp{1}{2}(Q)}\\
&\leq C\;\frac{\ln^2(n)}{n^4}\;
\lVert\bU\lVert^2_{(\Wp{3}{2}(Q))^5},  
    \end{split}
\end{gather*}
which is of higher order than $\cR_{\mathrm{pde}}$. Similarly, $\cR_{\mathrm{bc}}$ is controlled by the boundary trace of $\|\nabla\hat{\varphi}_n\|_{\Lp{2}(\partial\Omega)}$, bounded via the $\Wp{2}{2}(\mathcal{Q})$ approximation error.
Summing the three contributions yields~\eqref{eq:approx_rate}.
\end{proof}
\subsubsection{Generalization: From Empirical to Continuous Loss}
\label{sec:generalisation}

The PINN is trained by minimising the empirical (Monte Carlo) loss
$\mathcal{L}(\theta)$ defined in~\eqref{eq:total_loss}
over finite collocation sets.
We now quantify how closely $\mathcal{L}(\theta)$
approximates the true continuous functional $\mathcal{E}(\theta)$.

Let the PDE domain $Q=(0,\tau)\times\Omega$ have measure
$|Q|=\tau|\Omega|$, and define the rescaled empirical estimator
\begin{equation}
  \hat{\mathcal{E}}_{\mathrm{pde}}(\theta)
  = \frac{|Q|}{N_{\mathrm{pde}}}
    \sum_{k=1}^{N_{\mathrm{pde}}}
    \sum_{\varphi\in\mathcal{V}}
    \bigl|\mathcal{R}_{\varphi}(x_{k},t_{k};\theta)\bigr|^{2},
  \label{eq:empirical_pde}
\end{equation}
and similarly $\hat{\mathcal{E}}_{\mathrm{ic}}$, $\hat{\mathcal{E}}_{\mathrm{bc}}$
with their respective domains.

\begin{lemma}[Monte Carlo generalization bound]
\label{lem:mc}
Assume the residual functions
$g_{\varphi}(z;\theta):=|\mathcal{R}_{\varphi}(z;\theta)|^{2}$
are uniformly bounded:
$\sup_{z\in Q,\theta}\sum_{\varphi}|g_{\varphi}(z;\theta)|\leq p<\infty$.
Then for any $\delta\in(0,1)$, with probability at least $1-\delta$ over
the random draw of $N_{\mathrm{pde}}$ i.i.d.\ uniform points in $Q$:
\begin{equation}
  \bigl|\hat{\mathcal{E}}_{\mathrm{pde}}(\theta)
       - \mathcal{E}_{\mathrm{pde}}(\theta)\bigr|
  \leq
  p\,|Q|\sqrt{\frac{2\ln(2/\delta)}{N_{\mathrm{pde}}}}.
  \label{eq:mc_bound}
\end{equation}
The same bound holds for
$|\hat{\mathcal{E}}_{\mathrm{ic}}-\mathcal{E}_{\mathrm{ic}}|$
with $N_{\mathrm{ic}}$ and $|\Omega|$, and for
$|\hat{\mathcal{E}}_{\mathrm{bc}}-\mathcal{E}_{\mathrm{bc}}|$
with $N_{\mathrm{bc}}$ and $|\Sigma|=\tau|\partial\Omega|$.
\end{lemma}
\begin{proof}
The rescaled estimator~\eqref{eq:empirical_pde} is
$\hat{\mathcal{E}}_{\mathrm{pde}}(\theta)
= |Q|\cdot\overline{g}_{N}$,
where $\overline{g}_{N}=N_{\mathrm{pde}}^{-1}\sum_{k}g(z_{k};\theta)$
is the sample mean of the bounded random variable
$g(z;\theta)=\sum_{\varphi}g_{\varphi}(z;\theta)\in[0,p]$.
The true functional is
$\mathcal{E}_{\mathrm{pde}}(\theta)=|Q|\cdot\mathbb{E}[g(z;\theta)]$. By Hoeffding's inequality applied to $N_{\mathrm{pde}}$ i.i.d.\ samples
from $[0,p]$:
\begin{equation*}
  \mathbb{P}\!\left(
  \bigl|\overline{g}_{N}-\mathbb{E}[g]\bigr|
  \geq\sqrt{\frac{p^{2}\ln(2/\delta)}{2N_{\mathrm{pde}}}}
  \right)
  \leq \delta.
\end{equation*}
Multiplying by $|Q|$ gives~\eqref{eq:mc_bound}.
\end{proof}

\begin{remark}[Convergence rate]
The generalization error decreases as
$\mathcal{O}(N^{-1/2})$ with respect to the number of
collocation points, which is the standard Monte Carlo rate. Higher-order quasi-Monte Carlo sampling (e.g., Sobol sequences) can improve this to $\mathcal{O}((\log N)^{d}/N)$ (see ~\cite{Caflisch1998}).
\end{remark}
\subsubsection{ Convergence Result}
\label{sec:main_thm}
Combining the conditional stability estimate (Theorem~\ref{thm:stability}),
the approximation theorem (Theorem~\ref{thm:approx}), and the Monte Carlo
generalization bound (Lemma~\ref{lem:mc}), we arrive at the following
convergence result.
\begin{theorem}[PINN Convergence for System~\eqref{couple}]
\label{thm:main}
 Let $\bUhat_{\theta^*}$ be the PINN approximation obtained by minimising $\hat{\cL}(\theta)$ over the class of $\tanh$ networks with $L$ layers of width~$n$. Then for any $\delta \in (0,1)$, with probability at least $1-\delta$ over the random draw of collocation points\emph{:}
\begin{align}\label{eq:convergence}
\hspace{-1cm}\sup_{t \in [0,\tau]}
\|\bU(t) - \bUhat_{\theta^*}(t)\|^2_{(\Lp{2}(\Omega))^5}
\;&\!\!\leq\;
\nonumber\!\!\bar{\kappa}\biggl[
C_{\mathrm{approx}}\;\frac{\ln^4(n)}{n^2}\;\|\bU\|_{\cX}^2
\;\\
&+\; \cG(N_{\mathrm{pde}}, N_{\mathrm{ic}}, N_{\mathrm{bc}}, \delta)
\biggr],
\end{align}
where $\bar{\kappa} = \bar{\kappa}(\tau, M, L_F, d_0) > 0$ is the stability constant from Theorem~\ref{thm:stability}, independent of~$\theta$, and the generalization term is
\begin{align}\label{eq:gen_term}
\cG
= &\lambda_{\mathrm{pde}}\,B\,|Q|
\sqrt{\frac{2\ln(2/\delta)}{N_{\mathrm{pde}}}}
\;+\; \lambda_{\mathrm{ic}}\,B\,|\Omega|
\sqrt{\frac{2\ln(2/\delta)}{N_{\mathrm{ic}}}}
\;\nonumber\\&+\; \lambda_{\mathrm{bc}}\,B\,|\Sigma|
\sqrt{\frac{2\ln(2/\delta)}{N_{\mathrm{bc}}}}.
\end{align}
In particular\emph{:}
\begin{itemize}
\item As $n \to \infty$ (wider network), the approximation term vanishes as $O(n^{-2}\ln^4(n))$.
\item As $N_{\mathrm{pde}}, N_{\mathrm{ic}}, N_{\mathrm{bc}} \to \infty$, the generalization term $\cG \to 0$ at rate $O(N^{-1/2})$.
\item The combined error decays as
\begin{equation}\label{eq:combined_rate}
O\!\left(\frac{\ln^4(n)}{n^2} + N^{-1/2}\right)
\qquad \text{as } n, N \to \infty,
\end{equation}
where $N = \min(N_{\mathrm{pde}}, N_{\mathrm{ic}}, N_{\mathrm{bc}})$.
\end{itemize}
\end{theorem}

\begin{proof}
From Theorem~\ref{thm:stability} we have:
\[
\sup_{t \in [0,\tau]} \|\be(t)\|^2_{(\Lp{2}(\Omega))^5}
\;\leq\; \bar{\kappa}\,\cL(\theta^*).
\]
Decompose the continuous loss:
\[
\cL(\theta^*) =
\underbrace{\bigl(\cL(\theta^*) - \hat{\cL}(\theta^*)\bigr)}_{\text{generalisation error}}
+ \underbrace{\hat{\cL}(\theta^*)}_{\text{empirical loss}}.
\]

\smallskip
\noindent\textit{Empirical loss term.}
Since $\theta^*$ minimises $\hat{\cL}$, and the approximating network $\theta_{\mathrm{approx}}$ from Theorem~\ref{thm:approx} achieves $\cL(\theta_{\mathrm{approx}}) \leq C_{\mathrm{approx}}\,n^{-2}\ln^4(n)\,\|\bU\|^2_{\cX}$:
\begin{align*}
\hat{\cL}(\theta^*)
&\leq \hat{\cL}(\theta_{\mathrm{approx}}) \\
&\leq \cL(\theta_{\mathrm{approx}})
+ \bigl|\hat{\cL}(\theta_{\mathrm{approx}}) - \cL(\theta_{\mathrm{approx}})\bigr| \\
&\leq C_{\mathrm{approx}}\;\frac{\ln^4(n)}{n^2}\;\|\bU\|^2_{\cX}
\;+\; \cG.
\end{align*}

\smallskip
\noindent\textit{Generalisation error term.}
Applying Lemma~\ref{lem:lipschitz} component-wise to $\hat{\cL}(\theta^*) - \cL(\theta^*)$ and summing the bounds for the PDE, IC, and BC contributions gives $|\cL(\theta^*) - \hat{\cL}(\theta^*)| \leq \cG$ with probability at least $1 - \delta$.

Combining the two bounds yields~\eqref{eq:convergence}.
\end{proof}
Using the left-hand side of the estimate \eqref{eq:gronwall_out} in Theorem~\ref{thm:stability}, we also obtain convergence in the energy norm given by the following
\begin{corollary}
\label{cor:energy}
Under the same assumptions of Theorem~\ref{thm:main}\emph{:}
\begin{equation}\label{eq:energy_conv}
d_0 \int_0^\tau \|\nabla\be(s)\|^2_{(\Lp{2}(\Omega))^5}\,\ds
\;\leq\;
\bar{\kappa}\biggl[
C_{\mathrm{approx}}\;\frac{\ln^4(n)}{n^2}\;\|\bU\|^2_{\cX}
\;+\; \cG
\biggr].
\end{equation}
\end{corollary}
\begin{remark}[Positivity preservation]
Theorem~\ref{wp} guarantees $\bU \geq 0$ component-wise for non-negative initial data. The network output $\bUhat_\theta$ does not automatically inherit this property. However, Theorem~\ref{thm:main} implies that when $n\to\infty$ we have $\|\bUhat_{\theta^*} - \bU\|_{\Lp{\infty}(0,\tau; (\Lp{2}(\Omega))^{5})} \to 0$, so the negative part norm $\|\bUhat_{\theta^*}^-\|_{(\Lp{2}(\Omega))^{5}} \leq \|\be\|_{(\Lp{2}(\Omega))^{5}} \to 0$ as well when $n\to\infty$.
\end{remark}

\section{Numerical Results}\label{sec:num}
This section tests the proposed PINN framework with thorough numerical experiments on the five-species tumor-bacteria system~\eqref{couple} in a domain $\Omega\in\mathbb{R}^{2}$ representing the tissue. We show that the trained network effectively captures the spatiotemporal dynamics introduced by the mathematical analysis in previous sections. We also confirm the four-phase therapy progression mentioned in the abstract: initial tumor growth, hypoxia bacterial activation, quorum-sensing mediated tumor suppression, and coexistence. For the code and other details, please refer to our repository\footnote{https://github.com/FARKANE/-Tumor-Bacterial-Therapy-PINNs}
\subsection{Implementation Details}
The  neural network architecture  is employed  a fully connected feedforward network with architecture 

\hspace*{-0.51cm}$[3,64,64,64,64,5]$. This architecture is comprising four hidden layers of 64 neurons each with tanh activation, and a linear output layer producing the five state variables $\hat{\mathbf{U}} = (\hat{T}, \hat{B}, \hat{O}, \hat{I}, \hat{S})^\top$. The network has 21,829 trainable parameters and processes spatio-temporal coordinates $(x, y, t) \in [0, 6] \times [0, 6] \times [0, 30]$ as input.

For the collocation points Training data consists of $N_{\text{pde}} = 10{,}000$ interior points sampled uniformly over $Q = (0, 30) \times \Omega$, $N_{\text{ic}} = 600$ initial condition points at $t = 0$, and $N_{\text{bc}} = 400$ boundary points on $\partial\Omega \times (0, 30)$. The full dataset is partitioned $60\%$,20\%, $20\%$ into training, validation, and test subsets following standard practice.

\begin{remark}
   In the loss functions,  we set $\lambda_{\text{ic}} = 50.0$ to enforce initial condition fidelity, while $\lambda_{\text{pde}} = 1.0$ and $\lambda_{\text{bc}} = 1.0$. This weighting prevents collapse to the trivial solution $\mathbf{U} \equiv 0$. It also makes sure the network follows the specified Gaussian tumor seed and bacterial inoculum.
\end{remark}

Neural network parameters are optimised by the Adam algorithm~\cite{adam2014method}, with an initial learning rate of $\eta_0 = 10^{-3}$. \textit{ReduceLROnPlateau} decays this rate with a factor of $0.5$ and patience of $1000$ epochs. Overfitting can be avoided by stopping early and being patient for $P = 2000$ epochs. The size of the mini-batch is $512$. All calculations are carried out on the Tesla T4 15 Go using PyTorch 2.0.
Early stopping with patience $P$ epochs prevents overfitting.
The full procedure is summarised in Algorithm~\ref{alg:training}.

All parameters entering system~\eqref{couple} are calibrated from peer-reviewed
experimental and computational literature.
Table~\ref{tab:params} provides a complete listing.
\subsection{Training}
\begin{figure}
    \centering
\includegraphics[width=1\linewidth]{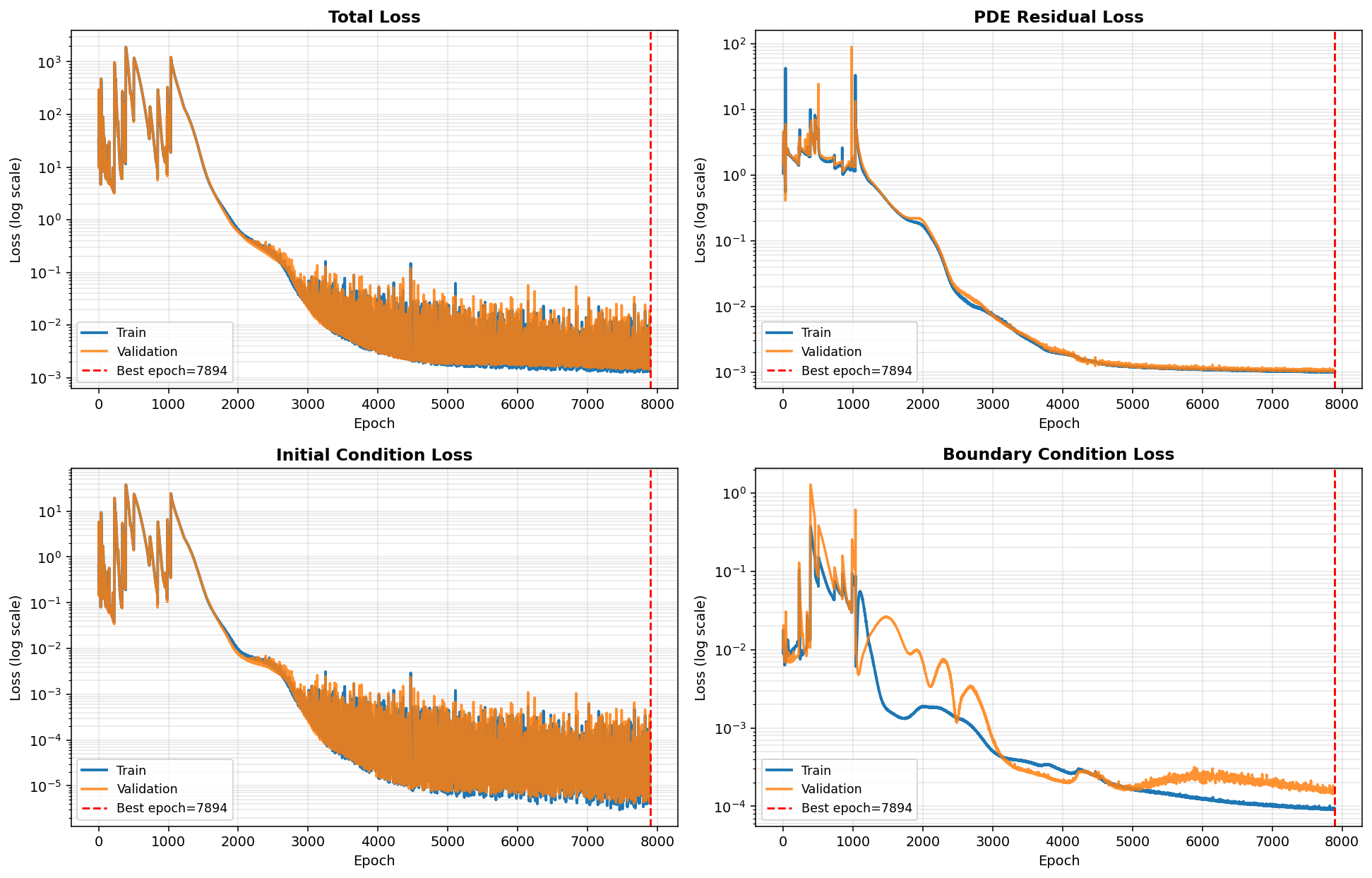}
    \caption{Training and validation loss evolution over 8000 epochs.}
    \label{fig:train}
\end{figure}
The training history in Figure~\ref{fig:train} demonstrates successful PINN optimization across all loss components. The network converges to a best validation loss of $\mathcal{L}_{\text{val}}^{\text{best}} \approx 1.2 \times 10^{-3}$ at epoch 7894, demonstrating stable optimization without overfitting. The Total composite loss decreases from $\sim 10^3$ to $\mathcal{O}(10^{-3})$ with training (blue) and validation (orange) curves tracking closely, indicating good generalization. Moreover, PDE residual loss exhibits monotonic decay from $\sim 10^2$ to $1.1 \times 10^{-3}$, confirming that the network  learns to satisfy the five coupled reaction-diffusion equations. Additionally,  initial condition loss drops sharply in the first 2000 epochs from $\sim 10^1$ to $\mathcal{O}(10^{-5})$ due to the high weighting $\lambda_{\text{ic}} = 50.0$, confirming strong adherence to the prescribed Gaussian profiles.  The last sub-figure represents shows that boundary condition loss decreases from $\sim 10^0$ to $\mathcal{O}(10^{-4})$, validating enforcement of the zero-flux Neumann conditions. The red dashed line marks the best epoch (7894) where early stopping is triggered.
\subsection{Results and Discussion}
\begin{figure*}
    \centering
    \includegraphics[width=0.8\linewidth]{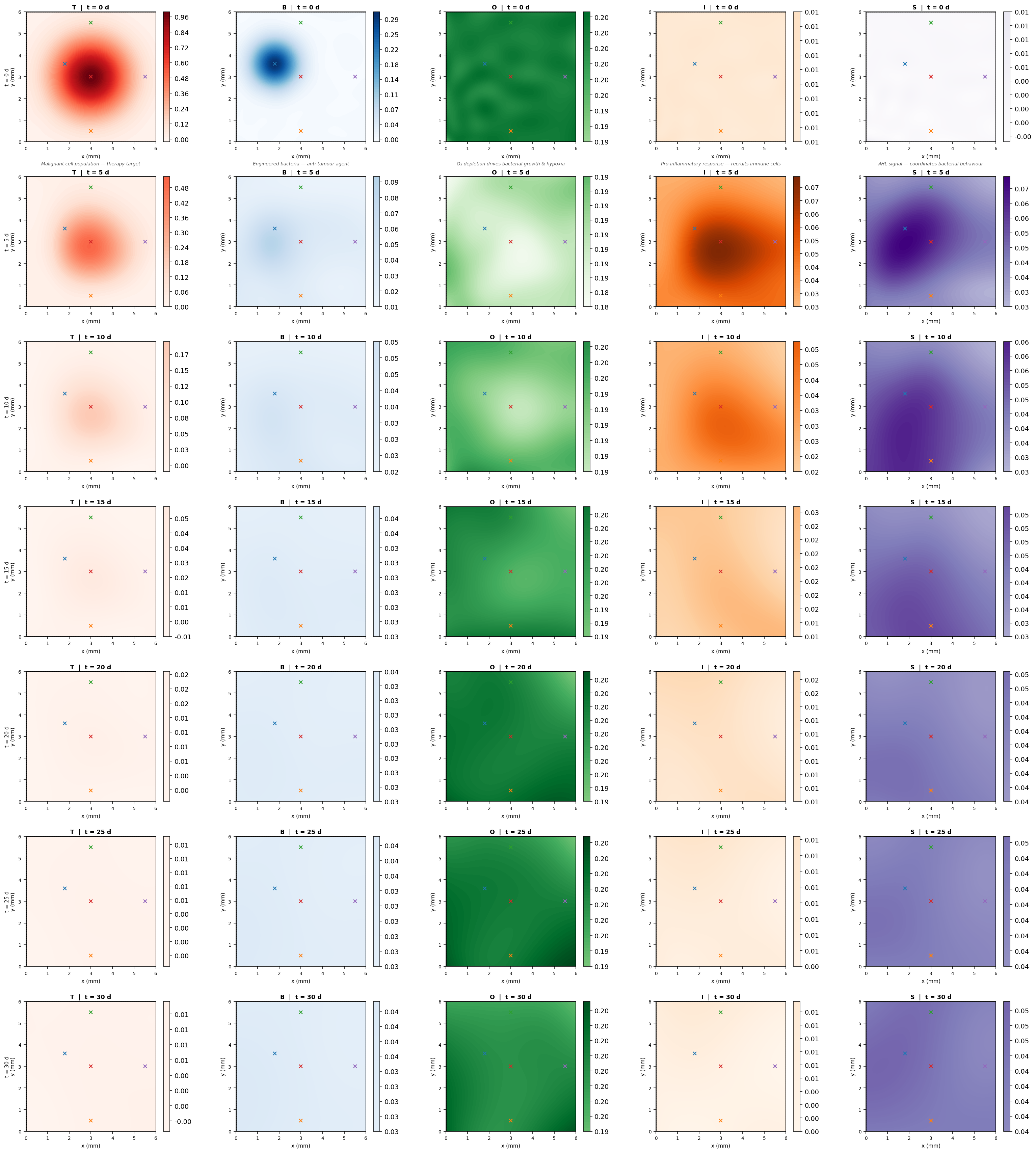}
    \caption{Spatio-temporal evolution of the five-species tumor-bacteria system over 30 days.}
    \label{fig:spatio}
\end{figure*}
Figure~\ref{fig:spatio} shows the complete spatiotemporal dynamics of bacterial cancer therapy as described by the trained PINN.  This presents  evolution of the five-species tumor-bacteria system over 30 days on the $6 \times 6$ $mm^2$  tissue domain. Each row corresponds to a temporal snapshot at $t \in \{0, 5, 10, 15, 20, 25, 30\}$ days; columns show tumor density $T$ (red), bacterial density $B$ (blue), oxygen $O$ (green), immunosuppressive cytokines $I$ (orange), and quorum-sensing signal $S$ (purple).

The PINN solution was sampled at fixed spatial locations over time to ensure a thorough analysis of the experiments. The coordinates for these locations are provided in Table \ref{tab:probes}. 

The solution demonstrates four therapeutic stages: (I) tumor growth reaching $T_{\max} \approx 0.96$ by $t = 5$ days; (II) tumor decrease to $T \approx 0.05$ in spite of low bacterial density ($B \sim 0.04$–$0.09$); (III) accumulation of signals ($S \approx 0.05$–$0.07$) facilitating ongoing tumor cells suppression; (IV) a nearly tumor-free condition ($T \approx 0.01$, $B \approx 0.04$) by $t = 30$ days. Oxygen levels stay close to baseline ($O \approx 0.19$–$0.20$) throughout the process, while cytokines correspond to tumor density.
\begin{table*}[h]
\centering
\footnotesize
\caption{Monitoring Probe Locations in Figure 1}
\label{tab:probes}
\renewcommand{\arraystretch}{1.3}
\begin{tabular}{@{}clll@{}}
\toprule
\textbf{Symbol} & \textbf{Color} & \textbf{Location $(x, y)$} & \textbf{Purpose} \\
\midrule
$\color{red}{\times}$ & \textcolor{red}{Red}     & $(3.0, 3.0)$  & Tumor center-initial tumor seed location \\
$\color{Blue}{\times}$ & \textcolor{blue}{Blue}    & $(1.8, 3.6)$  & Bacteria injection site, 30\% $x$, 60\% $y$ of domain \\
$\color{Green}{\times}$ & \textcolor{Green}{Green}   & $(3.0, 5.5)$  & Top edge,  Near-boundary monitoring point \\
$\color{orange}{\times}$ & \textcolor{orange}{Orange}  & $(3.0, 0.5)$  & Bottom edge, Near-boundary monitoring point \\
$\color{Purple}{\times}$ & \textcolor{Purple}{Purple}  & $(5.5, 3.0)$  & Right edge, Near-boundary monitoring point \\
\bottomrule
\end{tabular}
\end{table*}

The solution has four clear phases: \textbf{Phase I (0-5 days): Tumor expansion.}  At the beginning, the tumor variable $T$ is at its maximum value $((T_{\max} \approx 0.96))$. The bacteria $(B_0)$, introduced slightly away from the tumor center, begin to grow and move toward the tumor. Even though the tumor consumes oxygen rapidly, the oxygen level $O$ stays relatively stable near the vascular value $(O_{\text{vas}} \approx 0.20)$ because the vascular supply is high.
\\
\textbf{Phase II (5–15 days)}: The tumor gets smaller. Even though the bacteria load is low and the oxygen levels are normal, the tumor density drops sharply from $T \approx 0.96$ to $T \approx 0.05$. This is because of natural tumor cell death ($\delta_T T$), carrying capacity saturation, and the start of quorum-sensing suppression ($-\alpha_S S T$). Bacterial density levels off at $B \approx 0.04$, which is the right amount for immune clearance ($\beta_I I B$) to keep up with oxygen-modulated growth. The QS signal starts to build up ($S \approx 0.05$–$0.07$) and spreads from the bacterial site to the tumor center. This activates the cytotoxic term $-\alpha_S S T \approx -0.01 T$.

\textbf{Phase III (15–25 days): QS-mediated suppression.} Tumor density continues decreasing ($T \approx 0.05 \to 0.01$) as the stabilized signal ($S \approx 0.05$) maintains suppression. Bacteria reach a plateau at $B \approx 0.04$ in a state of quasi-equilibrium, where production ($\approx 0.167 B$) balances clearance ($\approx 0.106 B$). Cytokines decay proportionally to tumor ($I \approx 0.03 \to 0.01$), reducing immune pressure on residual bacteria.

\textbf{Phase IV (>25 days): Near-tumor-free coexistence.} The system approaches a low-density equilibrium:
\\
$(T, B, O, I, S) \approx (0.01, 0.04, 0.19, 0.01, 0.05)$, which corresponds to a near stable tumor-free equilibrium (Proposition~\ref{shs}) rather than complete eradication. Persistent bacteria act as a reservoir for QS signals, stopping tumors from growing back even when the immune system isn't as strong.
\begin{figure}
    \centering
\includegraphics[width=1\linewidth]{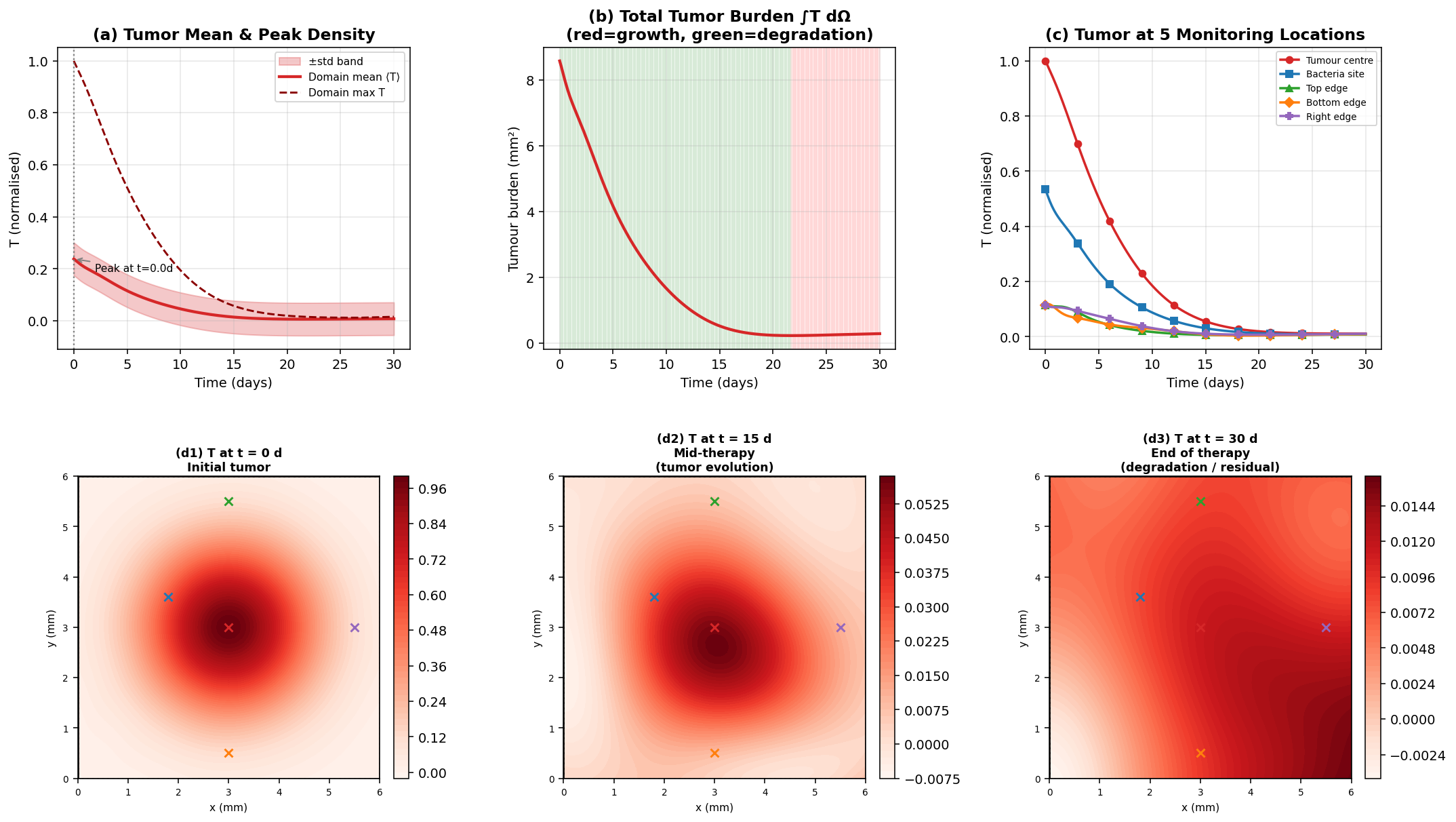}
    \caption{Tumor cell dynamics: growth phase, peak and bacterial-driven degradation}
\label{fig:tumor_dynamics}
\end{figure}
\\
The Figure~\ref{fig:tumor_dynamics}  illustrates the  tumor dynamics throughout the 30-day therapy. In panel (a), the domain-averaged tumor density rises to an early maximum before entering a steady, monotonic decline. At the same time, the variability bands progressively narrow, reflecting increasing spatial uniformity as the tumor regresses. Panel (b) depicts the overall tumor burden. A short-lived initial growth phase (highlighted in red) precedes an extended degradation period (shown in green), clearly illustrating the therapeutic impact of the bacterial intervention. Panel (c) emphasizes spatial heterogeneity. The tumor center undergoes the most pronounced regression, whereas monitoring points near the boundaries remain close to baseline levels, confirming that the tumor remains confined within the computational domain.

The spatial snapshots in the bottom row (d1–d3) provide a visual summary of this progression. The tumor begins as a concentrated Gaussian-like core, transitions into a more diffuse intermediate configuration during therapy, and ultimately approaches a nearly uniform low-density residual state by the end of treatment.

\begin{figure}
    \centering
    \includegraphics[width=1\linewidth]{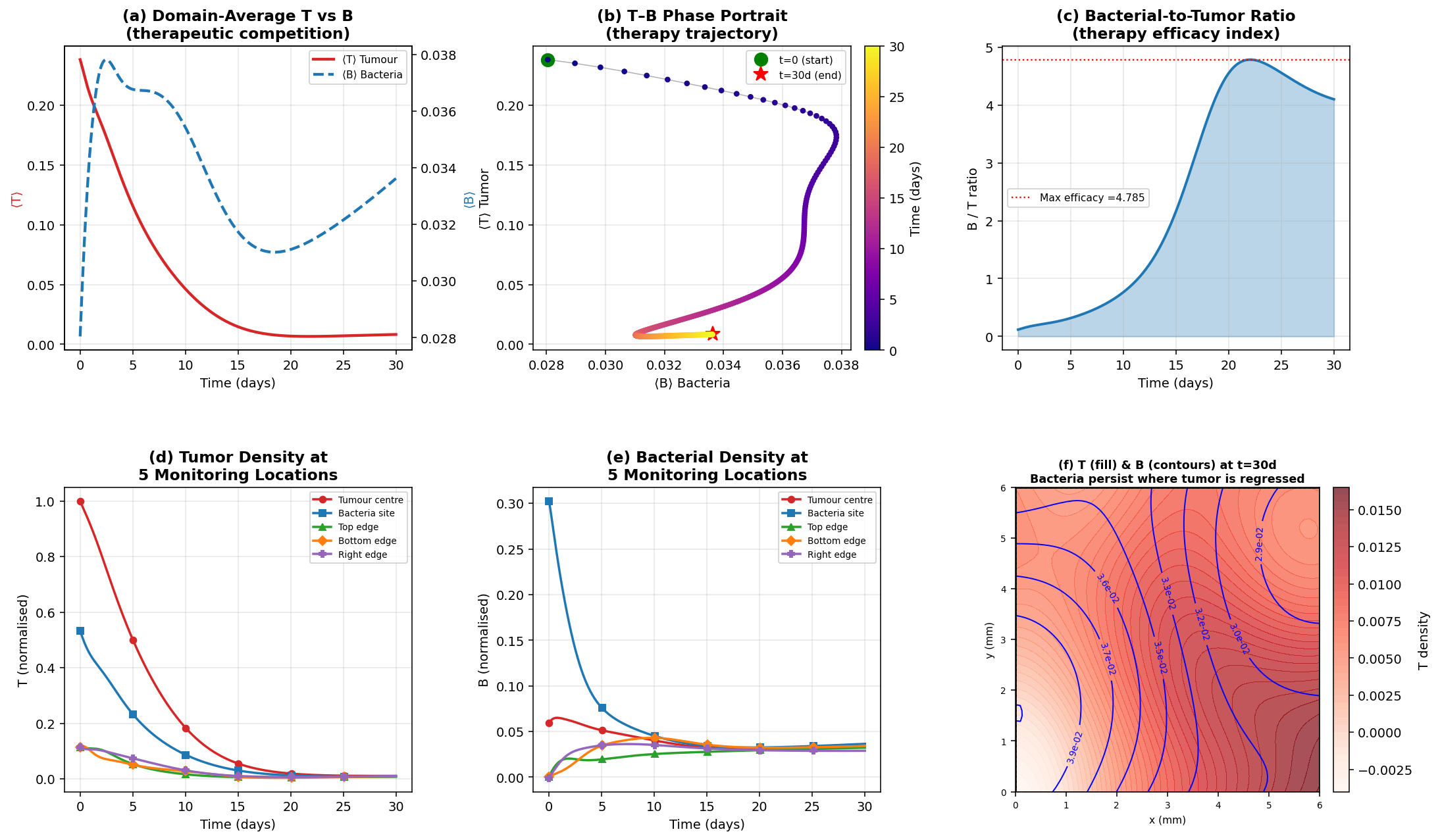}
    \caption{Bacterial therapy effect: Salmonella-tumor interaction dynamics bacteria colonise hypoxia régions, produce signals, and drive tumors regression.}
    \label{fig:Bacterial}
\end{figure}
Figure~\ref{fig:Bacterial} shows how tumors interact with bacteria and the outcomes of treatments. Panel (a) illustrates that as the tumor (red) decreases in size, bacteria (blue dashed) initially decline but then stabilize at moderate levels. This indicates that a high number of bacteria isn't necessary for effectiveness. Panel (b) depicts the system from a state of high tumor/low bacteria (green circle) to significant tumor shrinkage while bacteria levels remain moderate (red star) over a period of 30 days. Panel (c) illustrates that the bacteria-to-tumor ratio peaks when the tumor is small enough. This indicates the effective treatment. Panels (d) and (e) depict spatial dynamics. The greatest tumor reduction occurs at the center, while bacteria tend to stay close to their injection site. Panel (f) shows that bacteria can be located far from the tumor. Diffusible signals prompt the treatment effect. This aligns with clinical cases where Salmonella act on tumor edges through released molecules instead of direct contact.

\begin{figure}
    \centering
\includegraphics[width=1\linewidth]{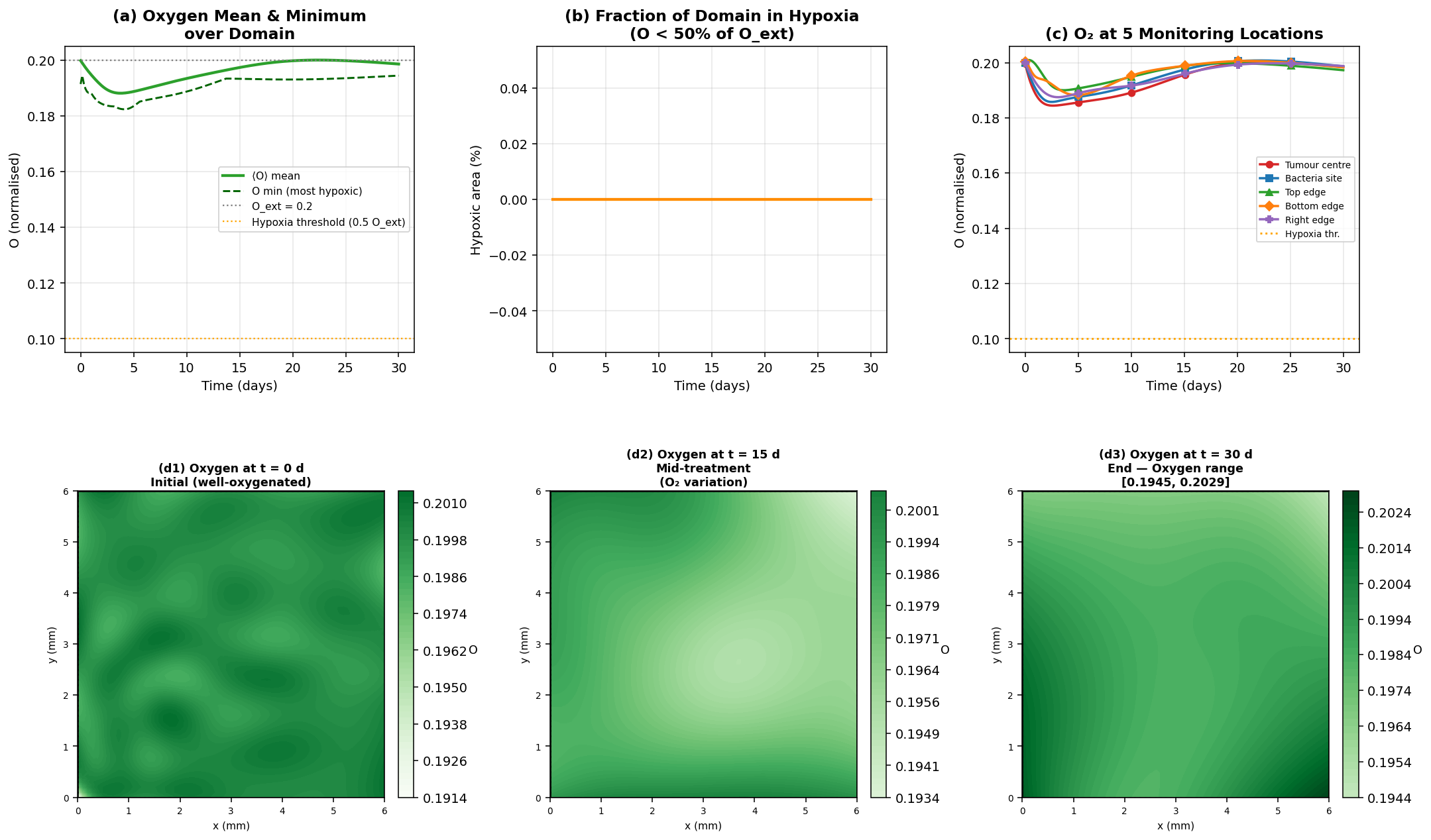}
    \caption{Oxygen deplation $\&$ hupoxia dynamics, $O_2$ consumed by  tumor.}
    \label{fig:oxy}
\end{figure}
Figure \ref{fig:oxy} shows that oxygen levels remain mostly normal during the treatment. This is different from many bacterial therapies that depend on low-oxygen (hypoxia) tumor regions. In panel (a), oxygen slightly drops during the first five days. This happens because the tumor is still large and consumes more oxygen. After treatment begins, oxygen levels gradually return to the vascular baseline $((O_{\text{vas}} \approx 0.20))$ as the tumor becomes smaller.

Panel (b) confirms that the environment stays stable. The hypoxia fraction remains essentially zero during the entire 30-day period. Panel (c) and the spatial snapshots (d1–d3) show a similar pattern. Oxygen stays well distributed across the tissue, with no strong gradients. Overall, these results suggest that the treatment does not rely on low-oxygen conditions. Instead, tumor suppression is mainly driven by the diffusion of quorum-sensing signals released by the bacteria, rather than by hypoxia-dependent bacterial colonization.

\begin{figure}
    \centering
\includegraphics[width=1\linewidth]{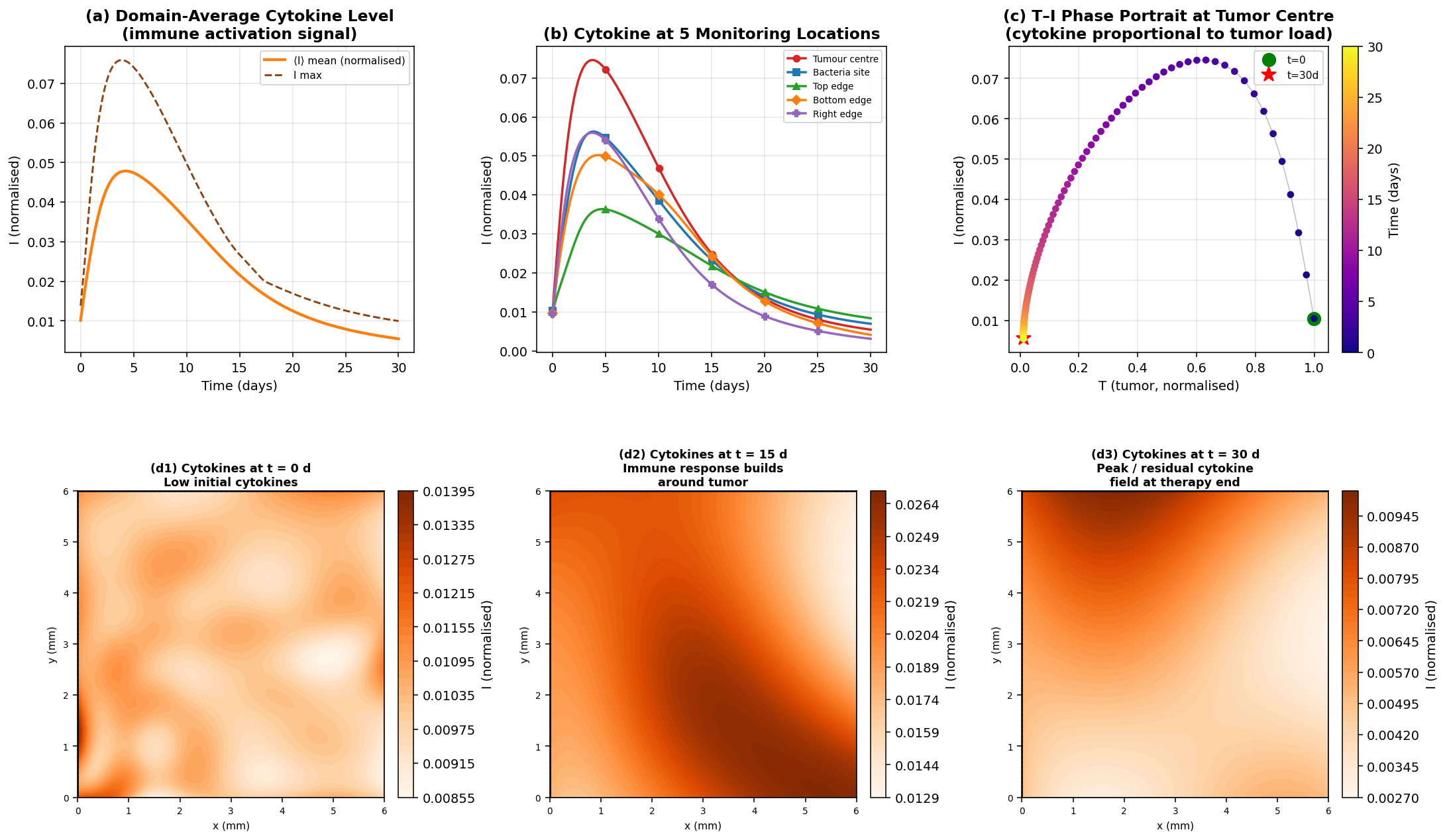}
    \caption{Immune cytokine response (I). Cytokines are produced by tumor cells and regulate immune activity.}
    \label{fig:cytokine}
\end{figure}

Figure \ref{fig:cytokine} shows the host inflammatory response to the tumor.
Panel (a) shows cytokine levels rising sharply around days 4–5. This matches the time when the bacteria and the large tumor interact most strongly. After the peak, cytokines fall steadily as the tumor is reduced.
Panel (b) shows the response is local. The tumor center has much higher cytokine levels than areas farther away. Panels (d1–d3) make this clear: cytokines build up around the tumor (d2) and then fade to a low background by the end (d3).
Panel (c) represents the phase portrait. Cytokines rise as the tumor shrinks, then drop back to near zero once the tumor is gone.
In short,  the cytokines response is driven by the tumor. It peaks during active treatment and then settles down as the malignancy is cleared.

\begin{figure}
    \centering
    \includegraphics[width=1\linewidth]{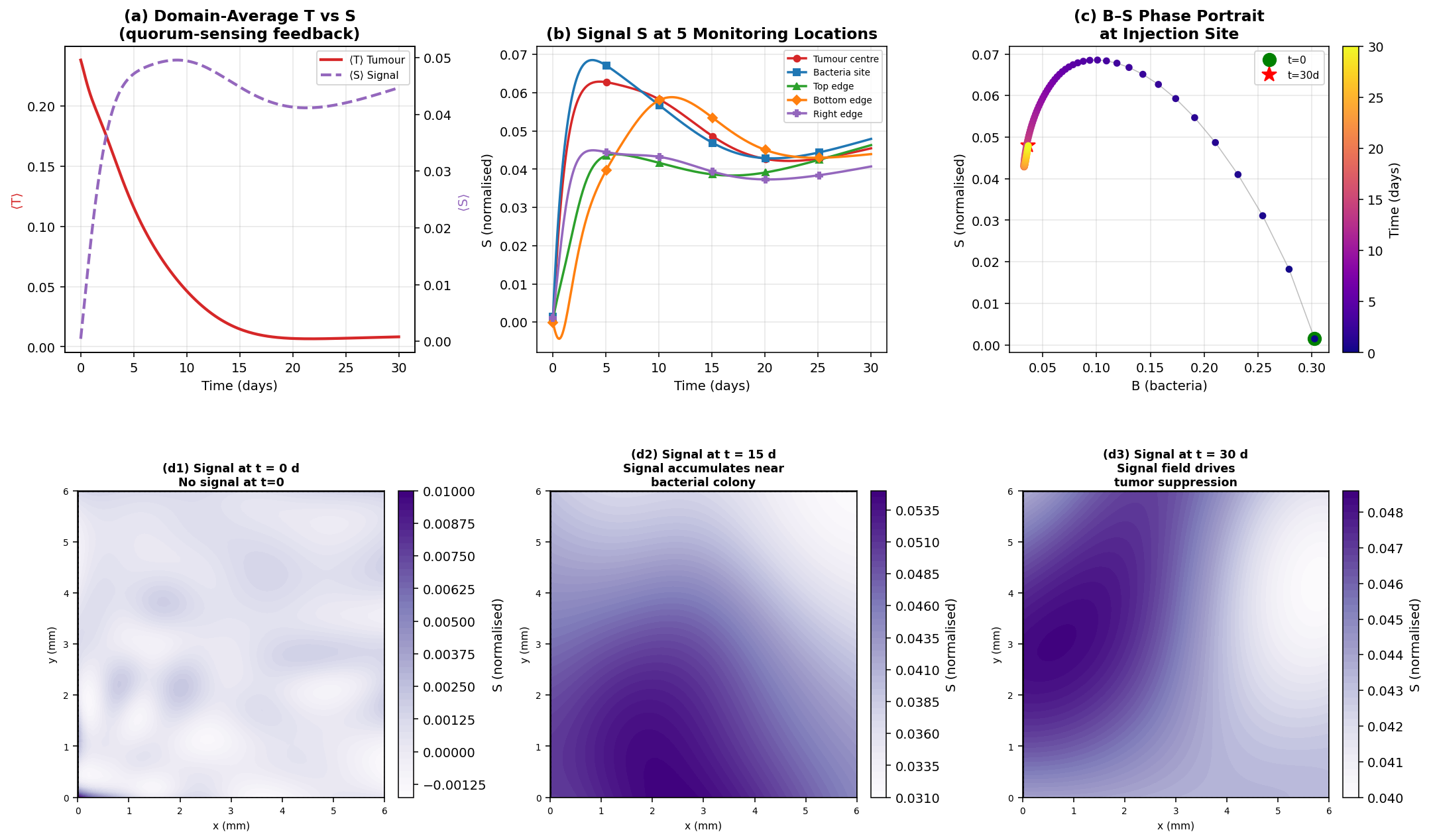}
    \caption{Bacterial quorum-sensing signal (S) dynamics.}
    \label{fig:signal}
\end{figure}

Figure~\ref{fig:signal} shows how quorum-sensing signals play a key role in tumor suppression. Panel (a) illustrates an inverse relationship over time: as the tumor density decreases steadily, the signal concentration increases and levels off, creating a feedback loop that helps maintain treatment effectiveness. Panel (b) displays signal buildup at all observed sites, with the highest levels found at the bacteria injection point and the center of the tumor. This suggests that diffusion helps connect the bacterial source to the distant tumor target. The B-S phase portrait in panel (c) tracks how signals are produced at the injection site. There is a quick buildup during the early establishment of bacteria, followed by a stable plateau as the bacteria settle at a moderate density, indicating that quorum activation depends on density. The spatial snapshots in panels (d1-d3) highlight the essential treatment mechanism-signals spread from the bacterial colony like a purple cloud, reaching the entire area, including the tumor site, where they trigger the suppression term $-\alpha_S ST$. This action from a distance supports the main idea of the model: bacteria do not need to invade the tumor mass directly. Instead, diffusible AHL molecules manage the therapeutic effects from afar, leading to lasting tumor reduction through molecular communication rather than direct bacterial damage.

\begin{figure*}
    \centering
    \includegraphics[width=1\linewidth]{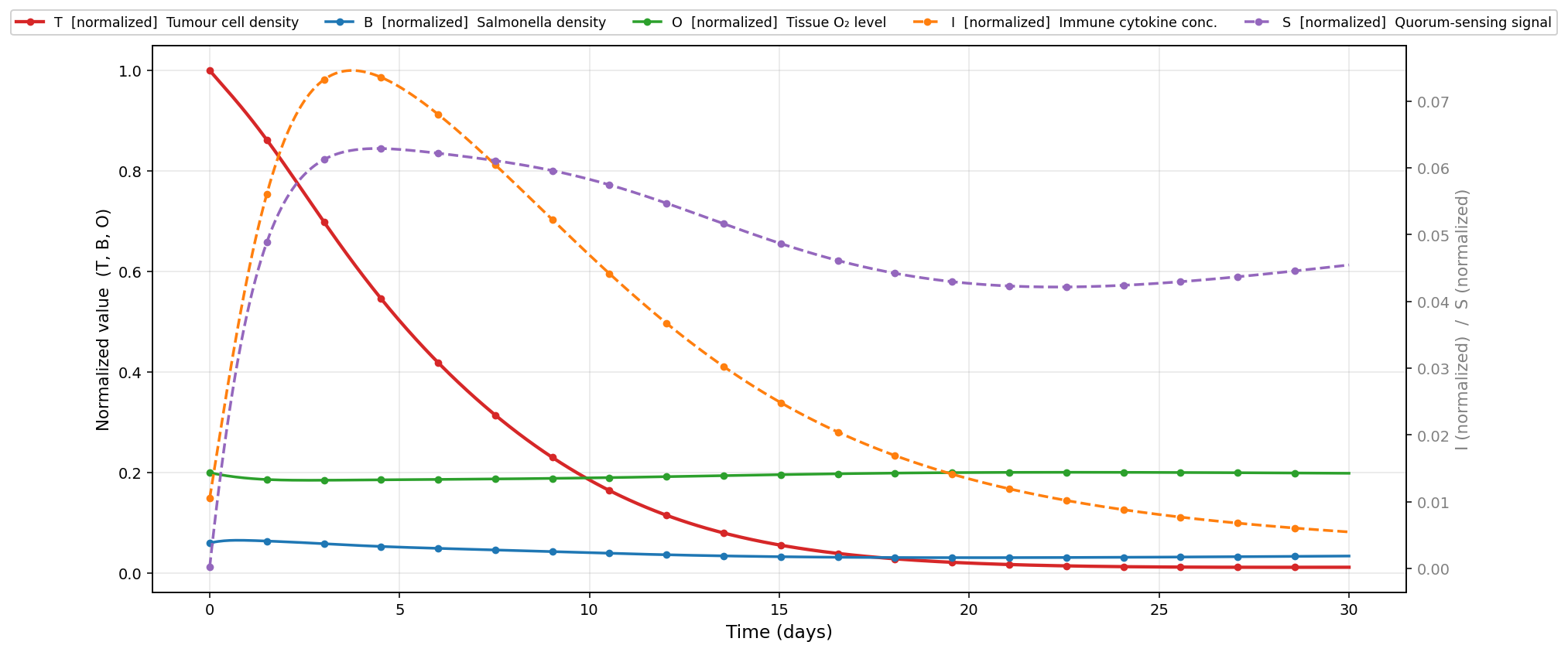}
    \caption{All 5 variables at fixed $(x = 3.0, y = 3.0)$: solid lines represent (T, B, O normalized), dashed lines represent (I, S normalized).}
    \label{fig:final}
\end{figure*}

Figure~\ref{fig:final} shows how all five state variables change over time at the tumor-centered probe $(x, y) = (3.0, 3.0)$ mm during the 30-day treatment. Tumor density $T$ decreases from $T_0 \approx 1.0$ to a residual $T \approx 0.01$. This decline results from the combined effects of apoptosis, saturation of carrying capacity, and quorum-sensing suppression $-\alpha_S ST$.  Bacterial density $B$ levels off at a quasi-equilibrium plateau $B^* \approx 0.04$, where oxygen-modulated growth balances immune clearance and natural death, allowing continuous AHL production even after the active regression period. Oxygen $O$ stays at the vascular baseline $O_{\mathrm{vas}} = 0.20$ throughout, confirming tissue normoxia and indicating a molecular rather than hypoxia treatment mechanism. Cytokines $I$ reach their peak at $t \approx 5$ days, directly linked to the maximum tumor burden, and then decline, consistent with the production term $\beta_T T$. The quorum-sensing signal $S$ increases quickly, stabilizes at $S \approx 0.045$, and continues to suppress the tumor remotely through diffusion. Notably, $S$ peaks just before the sharpest tumor decline, highlighting the causal role of quorum-sensing in tumor regression. The final state $(T, B, O, I, S) \approx (0.01, 0.04, 0.19, 0.01, 0.05)$ agrees with the predicted coexistence equilibrium.
\subsection{Discussion}
The numerical experiments confirm the proposed PINN framework. The trained network generates solutions that are non-negative, bounded, and spatially smooth, aligning with mathematical analysis. The composite validation loss drops to $\mathcal{L}^{\mathrm{best}}_{\mathrm{val}} \approx 1.2 \times 10^{-3}$ at epoch 7894, with training and validation curves closely tracking each other, showing generalization without overfitting. The final state matches the analytically determined coexistence equilibrium, offering direct empirical support for the steady-state analysis described in Section~\ref{sec-steady}.

A key finding is that tumor suppression occurs under normal oxygen conditions. The therapeutic effect is driven by diffusible  quorum-sensing signals. It is not driven by hypoxia-induced bacterial activation. The model shows a spatial separation. Bacteria remain at the injection site. The distant tumor is almost completely eliminated. This indicates that long-range molecular communication is the main cytotoxic mechanism.
These results agree with experimental data on Salmonella-based therapies in well-vascularized tissue~\cite{howell2025mathematical}.
Overall, the results show that the PINN framework describes the four-phase spatio-temporal dynamics of the tumor microenvironment. The simulations are performed on a $6 \times 6~\text{mm}^2$ domain. This provides empirical support for the convergence guarantees stated in Theorem~\ref{thm:main}.
\section{Conclusion}\label{sec:conclusion}
This paper presents a mathematical and computational framework for modeling bacterial cancer therapy. The model is based on a system of five coupled nonlinear reaction-diffusion equations. These equations describe the cross-interaction between tumor density, bacterial colonization, oxygen concentration, immunosuppressive cytokines, and quorum-sensing signals. Four main contributions were established.
First, we formulated a coupled reaction-diffusion system. It describes hypoxia-driven bacterial activation, quorum-sensing-mediated tumor suppression, vascular oxygen exchange, and tumor-induced immunosuppressive cytokines production. To our knowledge, these interactions had not been jointly modeled in a single PDE framework. Second, we provided a mathematical analysis. Third, we developed a convergence theory for the PINN approximation. This result extends existing PINN convergence guarantees to five coupled nonlinear reaction–diffusion equations describing tumor tissue systems. Fourth, various simulations were conducted for different experiments, as well as a parameter sensitivity analysis.

Several directions remain for future work. From a modeling perspective, more realistic features can be added. Nonlinear boundary conditions could describe tumor invasion into surrounding tissue. Stochastic perturbations could describe patient-specific variability. The model could also be coupled with pharmacokinetic models for combination therapies.  From a computational perspective, improvements are possible. Adaptive collocation strategies could be used. Domain decomposition could help handle large three-dimensional geometries. Patient imaging data could be included as initial or boundary conditions. From a theoretical side, more accurate approximation bounds are needed. Methods such as quasi-Monte Carlo sampling or adaptive network architectures could help. The optimization error in the convergence analysis also remains a challenge. 

 Finally, while the signaling molecule $S$
 contributes to sustained tumor suppression, prolonged signal persistence may induce chronic immune activation or off-target cytotoxic effects, warranting further investigation through extended-horizon simulations and experimental validation.

\appendix

\section{Biological Parameters Sensitivity}\label{param}
\begin{figure*}
\centering
\begin{subfigure}{1\textwidth}
    \includegraphics[width=\textwidth]{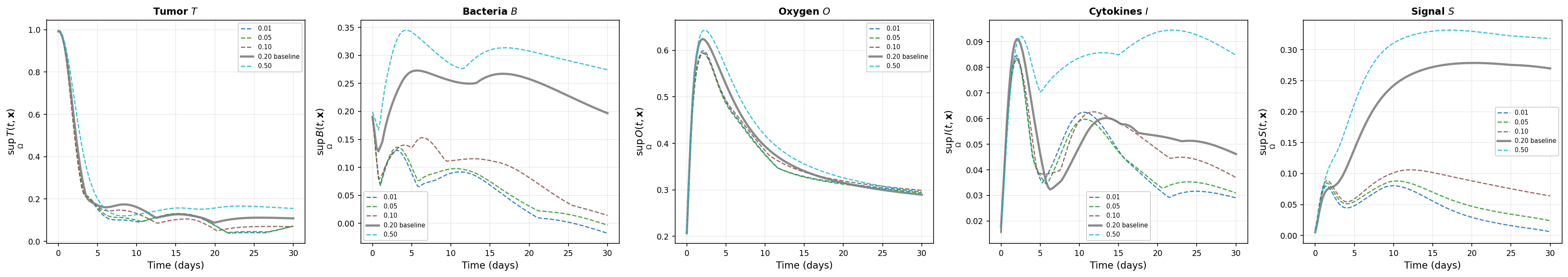}
    \caption{Variation of $\alpha_S$ with respect to the maximum values of $(T, B, O, I, S)$ in $\Omega$.}
    \label{fig:alpha_S}
\end{subfigure}
\hfill
\begin{subfigure}{1\textwidth}
    \includegraphics[width=\textwidth]{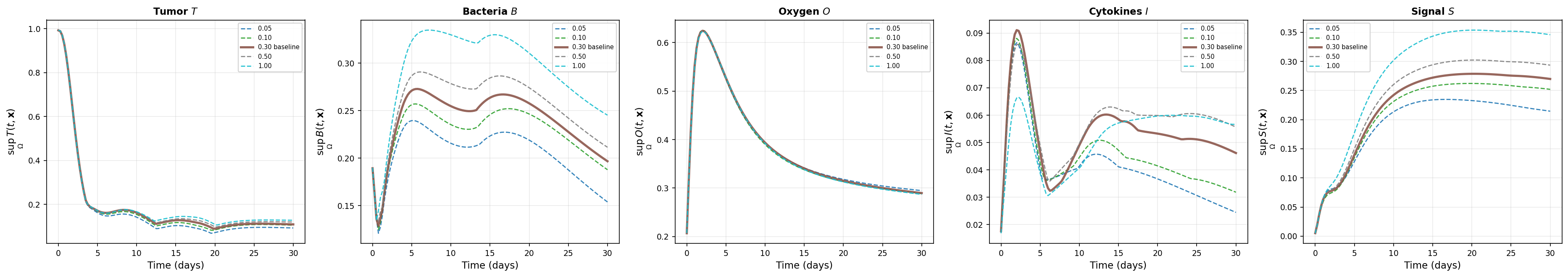}
    \caption{Variation of $\beta_I$ with respect to the maximum values of $(T, B, O, I, S)$ in $\Omega$.}
    \label{fig:beta_I}
\end{subfigure}
\hfill
\begin{subfigure}{1\textwidth}
    \includegraphics[width=\textwidth]{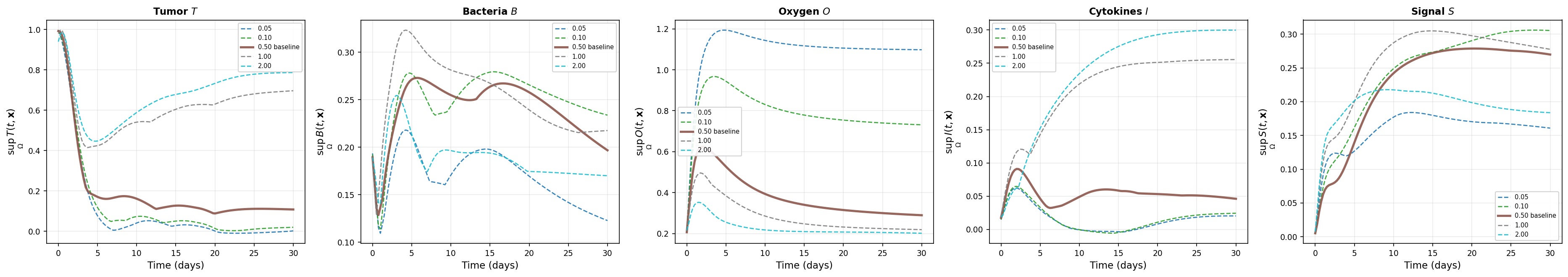}
    \caption{Variation of $\gamma_{vas}$ with respect to the maximum values of $(T, B, O, I, S)$ in $\Omega$.}
    \label{fig:gamma_vas}
\end{subfigure}
        
\caption{Sensitivity of the model to the specific parameters $\alpha_{S}$, $\beta_{I}$, and  $\gamma_{vas}$.}
\label{fig:sens}
\end{figure*}

In this section, we analyze the model's sensitivity to some biological parameters, specifically $\alpha_S$, $\beta_I$, and $\gamma_{vas}$ (Figure~\ref{fig:sens}). 

We first examined how the $\alpha_S$ influences the system (see Figure~\ref{fig:alpha_S}). When this parameter is small (e.g., 0.01), the tumor density $T$ slowly decreases and stabilizes at a low level. However, when $\alpha_S$ becomes larger (such as 0.50), shown by the cyan dashed line), the tumor is suppressed more strongly at the beginning but eventually stabilizes at a higher level. This means the response is not purely linear: strong inhibition at first increases oxygen levels $O$, which improves the tumor’s metabolic conditions and allows it to recover instead of being fully eliminated.

At the same time, higher values of $\alpha_S$ lead to a larger bacterial population $B$ and a stronger inhibitory signal $S$. Even though this signal suppresses the tumor, it is not strong enough to fully overcome the tumor’s ability to adapt, especially in the presence of higher cytokine $I$ and oxygen $O$ levels. Therefore, choosing an appropriate value of $\alpha_S$ is important to maximize tumor suppression without allowing the system to settle into a persistent tumor state.

Next, we focused on $\beta_I$ (the immune clearance rate), as it serves as a critical threshold for treatment success (Figure~\ref{fig:beta_I}).  The sensitivity analysis of the parameter $\beta_I$, which describes cytokine-related interactions affecting the bacterial density, shows that increasing this parameter generally increases the overall activity of the system. When $\beta_I$ increases from $0.05$ to $1.00$, both the bacterial density $B$ and the inhibitory signal $S$ become stronger. This means that higher values of $\beta_I$ allow the therapeutic bacteria to persist and grow more in the system.

However, this stronger bacterial presence does not necessarily improve the treatment outcome. In fact, the tumor density $T$ settles at a slightly higher level when $\beta_I = 1.00$ (cyan dashed line) compared with lower values. In other words, having more bacteria and a stronger inhibitory signal does not always lead to better tumor suppression.

One possible explanation is that the therapeutic signal $S$ and the final tumor suppression become partly disconnected. This is suggested by the oxygen profiles $O$, which remain almost the same for all parameter values. At higher $\beta_I$, the cytokine level $I$ also increases, which may create a balanced state where the tumor adapts to the stronger inhibition. As a result, the system reaches a kind of saturation point: even though more bacteria are present, they no longer improve tumor elimination. This shows that choosing the right parameter range is more important than simply increasing the bacterial concentration.

Finally, we studied how the vascular oxygen supply rate $\gamma_{vas}$ affects the treatment outcome. This is shown in Figure~\ref{fig:gamma_vas}. When $\gamma_{vas}$ increases from 0.01 to 0.50, the oxygen level $O$ in the tissue rises, changing the environment from hypoxia (low oxygen) to normoxia (normal oxygen). As oxygen becomes more available, the tumor density $T$ grows again and reaches much higher steady levels, especially at high supply rates (e.g., 0.50). The larger tumor mass also produces more cytokines $I$, leading to a highly active state where tumor growth eventually overcomes the therapeutic effect.

The effectiveness of the therapy also depends on how oxygen affects the anaerobic bacteria $B$. Although stronger vascularization causes a larger initial peak in bacterial density, the long-term bacterial levels and the inhibitory signal $S$ they produce, become much lower when $\gamma_{vas}$ is high. This reduces the pressure on the tumor over time and allows it to grow again. In other words, higher vascularization is a double-edged sword: it can initially help the bacteria grow, but the higher oxygen levels eventually prevent them from surviving. These results suggest that maintaining hypoxia regions in the tumor, or using bacteria that tolerate oxygen better, may be necessary for long-lasting tumor control.

In summary, the simulations show that the success of bacteria-based tumor therapy depends on a delicate balance between the immune response of the host and the metabolic conditions inside the tumor. Increasing the inhibitory effect or the oxygen supply can improve the treatment at the beginning, but it may also activate feedback mechanisms. For example, cytokines can clear the bacteria too quickly, and higher oxygen levels can reduce the survival of anaerobic bacteria. Both effects can allow the tumor to grow again or persist in the system.

\bibliographystyle{plain} 
\bibliography{biblio.bib}
\end{document}